\documentclass[11pt,twoside]{article}

\usepackage{float}
\usepackage{graphicx}
\usepackage{epstopdf}
\usepackage{multirow}
\usepackage{threeparttable}
\usepackage{xcolor}
\usepackage{rotating}
\usepackage{booktabs}
\usepackage{pifont}
\usepackage{enumitem}
\usepackage{amsmath,amssymb,amsthm,mathtools}
\usepackage{stmaryrd}
\usepackage{bbm}
\usepackage{array}
\usepackage{tabularx}
\usepackage{microtype}
\usepackage{cite}
\usepackage[colorlinks,linkcolor=red,citecolor=blue]{hyperref}
\allowdisplaybreaks[4]

\newcommand{\F}{\mathbb F}
\newcommand{\supp}{{\rm supp}}
\newcommand{\rank}{{\rm rank}}
\newcommand{\wt}{{\rm wt}}
\newcommand{\Tr}{{\rm Tr}}
\newcommand{\id}{{\rm id}}
\newcommand{\row}{{\rm row}}

\newcommand{\ket}[1]{\lvert #1\rangle}
\newcommand{\bra}[1]{\langle #1\rvert}
\newcommand{\calH}{\mathcal H}
\newcommand{\C}{\mathcal{C}}
\newcommand{\calM}{\mathcal M}
\newcommand{\Rec}{{\rm Rec}}

\newcommand{\Hull}{\rm Hull}
\newcommand{\vect}[1]{\mathbf{#1}}
\newcommand{\mat}[1]{\mathbf{#1}}
\newcommand\dsb[1]{\llbracket #1 \rrbracket}

\usepackage[
  a4paper,
  top=0.85in,
  bottom=0.9in,
  left=0.95in,
  right=0.95in,
  headheight=12pt,
  headsep=12pt,
  footskip=24pt
]{geometry}
\usepackage[colorlinks, linkcolor=red, citecolor=blue]{hyperref}

\makeatletter
\renewcommand\title[1]{\gdef\@title{\reset@font\Large\bfseries #1}}
\renewcommand\section{\@startsection {section}{1}{\z@}
                                   {-3.5ex \@plus -1ex \@minus -.2ex}
                                   {2.3ex \@plus.2ex}
                                   {\normalfont\large\bfseries}}
\renewcommand\subsection{\@startsection{subsection}{2}{\z@}
                                     {-3ex\@plus -1ex \@minus -.2ex}
                                     {1.5ex \@plus .2ex}
{\normalfont\normalsize\bfseries}}
\renewcommand\subsubsection{\@startsection{subsubsection}{3}{\z@}
                                     {-2.5ex\@plus -1ex \@minus -.2ex}
                                     {1.5ex \@plus .2ex}
{\normalfont\normalsize\bfseries}}
\def\@runningauthor{}\newcommand{\runningauthor}[1]{\def\runningauthor{#1}}
\def\@runningtitle{}\newcommand{\runningtitle}[1]{\def\runningtitle{#1}}
\renewcommand{\ps@plain}{
\renewcommand{\@evenhead}{\footnotesize\scshape \hfill\runningauthor\hfill}
\renewcommand{\@oddhead}{\footnotesize\scshape \hfill\runningtitle\hfill}}
\g@addto@macro\bfseries{\boldmath}
\makeatother

\theoremstyle{plain}
\newtheorem{theorem}{Theorem}

\newtheorem{lemma}{Lemma}
\newtheorem{corollary}{Corollary}

\theoremstyle{definition}
\newtheorem{definition}{Definition}
\newtheorem{example}{Example}

\newtheorem{problem}{Problem}

\theoremstyle{remark}
\newtheorem{remark}{Remark}

\begin{document}
\begin{sloppypar}

\title{Entanglement-Assisted Quantum Locally Recoverable Codes: 
Characterizations, Bounds, and Constructions
\thanks{The work
of Yang Li and San Ling was supported by the Nanyang Technological University
under Research Grant 04INS000047C230GRT01.  The work of Gaojun Luo was
supported by the National Natural Science Foundation of China under Grant
12401690. {\em (Corresponding author: Gaojun Luo)}}}

\author{Yang Li, San Ling, Zhenliang Lu, Gaojun Luo, and Shixin Zhu
\thanks{
Yang Li is with the School of Physical and Mathematical Sciences, Nanyang
Technological University, 21 Nanyang Link, Singapore 637371, Singapore
(e-mail: yanglimath@163.com).
San Ling is with the School of Physical and Mathematical Sciences, Nanyang
Technological University, 21 Nanyang Link, Singapore 637371, Singapore
(e-mail: lingsan@ntu.edu.sg).
He is also with VinUniversity, Vinhomes Ocean Park, Hanoi 100000, Vietnam
(e-mail: ling.s@vinuni.edu.vn).
Zhenliang Lu is with the Department of Computer Science, City University of
Hong Kong, Hong Kong, China (e-mail: zhenliang.lu@cityu.edu.hk).
Gaojun Luo is with the School of Mathematics, Nanjing University of Aeronautics
and Astronautics, Nanjing 210016, China (e-mail: gaojun\_luo@nuaa.edu.cn).
He is the corresponding author.
Shixin Zhu is with the School of Mathematics, Hefei University of Technology,
Hefei 230601, China (e-mail: zhushixin@hfut.edu.cn).
}}

\date{}
\maketitle

\begin{abstract}
Quantum locally recoverable codes (qLRCs) allow a single qudit erasure to be corrected  
by accessing only a small number of other qudits. Standard CSS and Hermitian constructions, 
however, impose dual-containing or self-orthogonal constraints on the underlying classical codes, 
thereby restricting the well-structured classical LRCs (cLRCs) that can be used to construct qLRCs. 
To relax these constraints, we introduce entanglement-assisted quantum locally recoverable codes 
(EAQLRCs) by assuming that halves of the pre-shared maximally entangled pairs are noiseless. 
We characterize sufficient support conditions on extended stabilizers under which 
entanglement-assisted stabilizer codes have locality $r$ and derive a CSS-like construction 
from two classical codes without imposing the ordinary dual-containing condition. 
We further establish an upper bound on locality and a Singleton-like bound 
for arbitrary CSS-like EAQLRCs, and characterize the pure codes attaining equality in the latter bound. 
These results yield a general framework for constructing optimal pure EAQLRCs from pairs of cLRCs. 
Applying this framework to $\ell$-intersection pairs of MDS codes and block parity-check matrices, 
we obtain two families of optimal pure CSS-like EAQLRCs with flexible parameters and nontrivial localities. 
To the best of our knowledge, these represent the first explicit families of EAQLRCs.
\end{abstract}

\noindent{\bf Keywords:} Locally recoverable code, entanglement-assisted quantum code,
CSS construction, $\ell$-intersection pair, block parity-check matrix

\noindent{\bf AMS Classification (MSC 2020)}: 94B05, 94B35, 81P45

\section{Introduction and motivation}

Throughout this paper, $q$ is a prime power, $\F_q$ is the finite field with $q$ elements, and $[m]=\{1,2,\ldots,m\}$ for every nonnegative integer $m$, with $[0]=\emptyset$.
An $[n,k,d]_q$ {\em linear code} $\C$ is a $k$-dimensional linear subspace of $\F_q^n$ with minimum distance $d=d(\C)$.
Given an $[n,k]_q$ linear code $\C$, its {\em Euclidean dual} is the $[n,n-k]_q$ linear code
$$\C^{\perp}=\left\{{\bf y}=(y_1,y_2,\ldots,y_n)\in\F_q^n:~\sum_{i=1}^n x_i y_i=0,\ \forall\,{\bf x}=(x_1,x_2,\ldots,x_n)\in\C\right\},$$
and its \emph{Euclidean hull} is the intersection code $\Hull(\C)=\C\cap\C^\perp$.
A linear code $\C$ is \emph{Euclidean self-orthogonal} if $\Hull(\C)=\C$, and \emph{Euclidean dual-containing} if $\Hull(\C)=\C^\perp$.

\subsection{Classical and quantum locally recoverable codes}

In classical and quantum error correction, an {\em erasure} means an error whose position is known \cite{GrasslBethPellizzari1997}.
Erasure errors are common failure modes in modern distributed storage systems (DSSs), where they correspond to failures of storage nodes.
In particular, the failure of a single storage node occurs more frequently than the simultaneous catastrophic failure of multiple nodes in large DSSs \cite{SharmaRamkumarTamo2024}.
To repair a single failed node efficiently, Gopalan $et~al.$ \cite{GopalanHuangSimitciYekhanin2012} introduced the notion of {\em classical locally recoverable codes (cLRCs)}.
Specifically, for each $i\in[n]$ and each codeword ${\bf c}=(c_1,c_2,\ldots,c_n)\in\C$, if the $i^{\rm th}$ symbol $c_i$ can be recovered by accessing 
at most $r$ other code symbols of ${\bf c}$, then the $[n,k,d]_q$ linear code $\C$ is called an 
{\em $(n,k,d,q;r)$-cLRC}, or a {\em cLRC with locality $r$}.
During the past decade, fundamental bounds and many optimal constructions of cLRCs have been obtained  
in \cite{GopalanHuangSimitciYekhanin2012,PrakashKamathLalithaKumar2012,PapailiopoulosDimakis2014,TamoBarg2014,RawatPapailiopoulosDimakisVishwanath2014,CadambeMazumdar2015,Jin2019LongLRC,ChenHuang2019OptimalLRC,XingYuan2022LongLRC} and the references therein.

With the rapid development of quantum computing and quantum storage, 
large-scale quantum data storage may become feasible in the future.
This motivates a quantum counterpart of cLRCs, namely {\em quantum locally recoverable codes (qLRCs)}.
Following the standard terminology of quantum error correction \cite{KnillLaflamme1997,LidarBrun2023QuantumErrorCorrection}, an $\dsb{n,\kappa,\delta}_q$ {\em quantum error-correcting code (QECC)} $\mathcal Q$ is a $q^\kappa$-dimensional subspace of $(\mathbb C^q)^{\otimes n}$ with minimum distance $\delta$.
It can correct every set of qudit erasures affecting at most $\delta-1$ known positions.
Golowich and Guruswami \cite{GolowichGuruswami2023} introduced qLRCs with locality $r$.
Informally, an $\dsb{n,\kappa,\delta}_q$ QECC $\mathcal Q$ has locality $r$ if, 
whenever any single qudit of a state $\ket{\varphi}\in\mathcal Q$ is erased, 
the state can be recovered by a recovery channel that accesses at most $r$ other qudits.
Thus, instead of invoking a global recovery procedure, a single qudit erasure can be corrected using a small quantum subsystem.
By restricting recovery to a small subsystem, qLRCs can reduce the number of qudits accessed 
and the number of measurements required to correct a single qudit erasure \cite{DZ2020,M2025}.
Golowich and Guruswami \cite{GolowichGuruswami2023} also established a connection between 
qLRCs and quantum low-density parity-check (qLDPC) codes, 
highlighting that qLRCs serve as a foundational step toward studying 
stronger locality properties in qLDPC codes.

Motivated by these applications and theoretical connections, qLRCs have attracted increasing attention.
Golowich and Guruswami \cite{GolowichGuruswami2023} and Galindo $et~al.$ \cite{GalindoHernandoMartinCruzMatsumoto2024} 
established useful links between the localities of QECCs derived from the CSS and Hermitian constructions \cite{KetkarKlappeneckerKumarSarvepalli2006} 
and those of the Euclidean and Hermitian dual-containing classical codes, respectively. 
These links give CSS-like and Hermitian-like construction criteria for obtaining qLRCs 
from cLRCs. 
The definition of qLRCs has also been generalized to qLRCs with 
$(r,\delta)$-locality \cite{GalindoHernandoMartinCruzMatsumoto2024}, 
hierarchical locality \cite{GuruswamiKshirsagarTrivedi2026}, 
and asymmetric locality \cite{LiJinXing2026Asymmetric}. 
Furthermore, the following advances have been made in the study of qLRCs.  
\begin{itemize}
  \item On the bound side, Golowich and Guruswami first established a Singleton-like bound for general qLRCs \cite{GolowichGuruswami2023}.
Subsequent work derived bounds for qLRCs from the CSS and Hermitian constructions, 
including Singleton-like, Griesmer-like, Cadambe-Mazumdar-like, Plotkin-like, 
and sphere-packing-like bounds, with separate refinements for pure codes \cite{LiLiLuoLing2025PureQLRC,GalindoHernandoMartinCruzMatsumoto2024} 
and general codes \cite{LuoChenEzermanLing2025,LiLiLaoLuoLing2026TDesign}.

  \item On the construction side, the CSS-like construction has been applied to obtain qLRCs from variations of hypergraph products \cite{BuGuLi2025}, Tamo-Barg codes \cite{GolowichGuruswami2023}, good polynomials \cite{SharmaRamkumarTamo2024}, parity-check matrices and cyclic codes \cite{LuoChenEzermanLing2025}, and trace codes \cite{XieZhuSun2025OptimalQLRC}.
The Hermitian-like construction has been developed using dual-containing MDS codes \cite{GalindoHernandoMartinCruzMatsumoto2024}, Hamming, generalized Reed-Muller, and Solomon-Stiffler codes \cite{LiLiLuoLing2025PureQLRC}, 
and codes supporting combinatorial designs \cite{LiLiLaoLuoLing2026TDesign}.
Further constructions use matrix-product codes \cite{GalindoHernandoMunueraRuano2025MatrixProduct,CaoZhou2026MatrixProduct,CaoZhou2026NSC,CaoZhou2026MP,ZhouCao2026Classical}, BCH and homothetic BCH codes \cite{GalindoHernandoMatsumoto2026BCH}, and impure stabilizer codes \cite{GalindoHernandoMartinCruzMatsumoto2026Impure}.

  \item On the optimality side, the available bounds are not generally attained for arbitrary parameters.
Most explicit families of qLRCs currently proved optimal are pure and 
are with respect to the Singleton-like bound 
specialized to the CSS-like or Hermitian-like construction 
\cite{CaoZhou2026NSC,GalindoHernandoMunueraRuano2025MatrixProduct,GalindoHernandoMartinCruzMatsumoto2024,GalindoHernandoMatsumoto2026BCH,CaoZhou2026MP,LuoChenEzermanLing2025,XieZhuSun2025OptimalQLRC,LiLiLaoLuoLing2026TDesign,LiLiLuoLing2025PureQLRC,CaoZhou2026MatrixProduct,ZhouCao2026Classical}.
\end{itemize}

\subsection{Motivation for introducing locality into entanglement-assisted QECCs}

Although the locality of a QECC is defined in terms of local recovery channels, 
the standard algebraic constructions of qLRCs cannot be applied to arbitrary cLRCs. 
For a CSS-like construction from two classical codes $\C_X,\C_Z\subseteq\F_q^n$ \cite{GolowichGuruswami2023}, one requires the dual-containing condition $\C_X^\perp\subseteq\C_Z$.
Similarly, the Hermitian-like construction \cite{GalindoHernandoMartinCruzMatsumoto2024} requires an appropriate Hermitian self-orthogonal or dual-containing cLRC.
These conditions guarantee that the corresponding quantum stabilizers commute, 
but they also restrict the cLRCs that can be used directly.
Indeed, cLRCs are usually designed to optimize locality, minimum distance, length, or dimension, 
and an optimal cLRC need not satisfy any dual-containing condition.
Consequently, requiring the dual-containing property excludes many optimal and well-known 
constructions of cLRCs and narrows the attainable parameter ranges of qLRCs. 

Entanglement-assisted QECCs (EAQECCs), introduced by Brun $et~al.$ \cite{BrunDevetakHsieh2006}, 
use pre-shared entanglement to remove the dual-containing constraint. 
In an EAQECC, Alice and Bob share $c$ maximally entangled qudit pairs in advance, 
with Bob's halves assumed to remain reliable. 
This shared entanglement permits noncommuting stabilizer checks, 
thereby replacing the self-orthogonality or dual-containing property 
by the consumption of maximally entangled pairs \cite{BrunDevetakHsieh2006,HsiehDevetakBrun2007}. 
Consequently, pairs of classical codes that do not satisfy dual-containing conditions 
can be directly employed to construct EAQECCs \cite{GalindoHernandoMatsumotoRuano2019}. 
Assume that a qudit has Hilbert space $\mathbb C^q$ with computational basis $\{\ket{x}:~x\in\F_q\}$.
For an ordered register $A=(A_1,A_2,\ldots,A_t)$ of qudits, write 
$\calH_A=\bigotimes_{j=1}^t\calH_{A_j}\cong(\mathbb C^q)^{\otimes t}$
for its joint Hilbert space, with $\calH_A=\mathbb C$ when $|A|=0$.
Let $Q_{[n]}=(Q_1,Q_2,\ldots,Q_n)$ be the transmitted register and 
let $B_{[c]}=(B_1,B_2,\ldots,B_c)$ be Bob's reliable register. 
When $n$ and $c$ are clear from the context, we abbreviate 
$Q_{[n]}$ and $B_{[c]}$ as $Q$ and $B$, respectively. 
Then an EAQECC can be defined as follows. 

\begin{definition}\label{def.eaqecc}
An $\dsb{n,\kappa,\delta;c}_q$ \emph{EAQECC} $\mathcal Q\subseteq\calH_{Q_{[n]}}\otimes\calH_{B_{[c]}}$ 
encodes $\kappa$ logical qudits into $n$ transmitted physical qudits with the help of 
$c$ maximally entangled pairs shared in advance by the sender Alice and the receiver Bob, 
and has minimum distance $\delta$, 
which means that $\mathcal Q$ can correct every set of qudit erasures affecting at most $\delta-1$ transmitted positions.
More precisely, $\mathcal Q$ is the image of the encoding isometry $V$ given by
$$
\begin{aligned}
V &: (\mathbb C^q)^{\otimes\kappa}\longrightarrow(\mathbb C^q)^{\otimes n}\otimes(\mathbb C^q)^{\otimes c}, \\
V\ket{\phi} &=\frac{1}{\sqrt{q^c}}\sum_{\vect{x}\in\F_q^c}U_{\rm enc}\left(\ket{\phi}\otimes\ket{0}^{\otimes(n-\kappa-c)}\otimes\ket{\vect{x}}\right)\otimes\ket{\vect{x}}_{B_{[c]}},
\end{aligned}
$$
where $0\le c\le n-\kappa$, $\ket{\phi}$ is a state of $\kappa$ logical qudits, $\ket{\vect{x}}=\bigotimes_{j=1}^c\ket{x_j}$ for $\vect{x}=(x_1,x_2,\ldots,x_c)\in\F_q^c$, and $U_{\rm enc}$ is a unitary operator acting on Alice's $n$ input qudits and producing the transmitted register $Q_{[n]}$.
\end{definition}

By Definition~\ref{def.eaqecc}, only the register $Q_{[n]}$ is transmitted through the noisy channel, while $B_{[c]}$ remains with Bob and is assumed to be noiseless. Since $U_{\rm enc}$ acts only on Alice's $n$ qudits and leaves Bob's register unchanged, the overall encoding operation is $U_{\rm enc}\otimes I_{B_{[c]}}$. Setting $c=0$ in Definition~\ref{def.eaqecc} recovers the standard definition of ordinary QECCs.
Therefore, the possibility of recovering an erased transmitted qudit by accessing only a small number of other transmitted qudits, 
together with the removal of the dual-containing constraint through entanglement assistance, 
naturally motivates the introduction of locality into EAQECCs. 
We call EAQECCs with such a local recovery property \emph{entanglement-assisted quantum locally recoverable codes (EAQLRCs)}. These considerations lead to the following problem.

\begin{problem}\label{prob.1}
How should locality be defined for EAQECCs, and how can optimal EAQLRCs be constructed?
\end{problem}

\subsection{Our contributions}

This paper answers Problem~\ref{prob.1}. 
We summarize our contributions as follows.

\begin{enumerate}
\item [\rm 1)] We model a single transmitted qudit erasure by the quantum channel in 
Definition~\ref{def.single-qudit-erasure-channel}, give an equivalent Pauli representation 
in Lemma~\ref{lem.single-qudit-erasure-channel}, 
and define recovery sets in Definition~\ref{def.ij-local}. 
Each recovery set indexes only transmitted qudits, 
while the recovery channel may also access Bob's reliable register. 
Based on the sizes of these recovery sets, we define EAQLRCs in Definition~\ref{def.eaqlrc}. 
For entanglement-assisted stabilizer codes (EASCs), 
we construct local recovery channels from extended stabilizers satisfying suitable 
support conditions, yielding a characterization of sufficient conditions such that 
an EASC has locality $r$ in Theorem~\ref{thm.stabilizer-locality}. 
Moreover, we give a CSS-like construction 
of EAQLRCs from two classical codes $\C_X$ and $\C_Z$ without imposing 
the ordinary dual-containing condition in Theorem~\ref{thm.eacss}. 
In this construction, the locality of the resulting EAQLRC depends on 
the joint supports of the codewords in $\C_X^\perp$ and $\C_Z^\perp$,  
whereas we show that the locality of a single classical code with 
$\C_X=\C_Z=\C$ transfers directly to the resulting EAQLRC in 
Corollary~\ref{cor.symmetric-eacss}.

\item  [\rm 2)]
For an arbitrary $\dsb{n,\kappa,\delta;c}_q$ CSS-like EAQLRC arising from 
Theorem~\ref{thm.eacss}, whether pure or impure, 
we establish the following two bounds for evaluating optimality.
\begin{itemize}
  \item We establish the upper bound
  $$r\le\min\{n-1,\dim(\mathcal C_X+\mathcal C_Z)\}$$
  in Theorem~\ref{thm.elementary-locality}, where $\C_X$ and $\C_Z$ 
  are the two underlying classical codes. 
  We call an EAQLRC whose locality attains this upper bound \emph{trivial}. 
  We further show that applying Corollary~\ref{cor.symmetric-eacss} 
  to an MDS code produces only a trivial EAQLRC in Remark~\ref{rem.trivial-MDS}.

  \item We establish the Singleton-like bound
  $$2\delta\le n+c-\kappa-2\left\lceil\frac{\kappa}{r}\right\rceil+4$$
  in Theorem~\ref{thm.lrc-bounds}. 
  For pure CSS-like EAQLRCs, we further present necessary and sufficient conditions 
  for attaining equality in this bound in Theorem~\ref{thm.singleton-equality}. 
  Based on these conditions, we develop a general framework for constructing optimal 
  pure CSS-like EAQLRCs in Subsection~\ref{subsec:framework},  
  where we begin with parity-check matrices $\mat H_X$ and $\mat H_Z$ 
  of suitable classical codes and use a nonsingular diagonal matrix $\mat D$ 
  to adjust the number of maximally entangled pairs $c$ and improve the possibility 
  that the resulting EAQLRCs have optimal parameters.
\end{itemize}

\item [\rm 3)] We apply this framework in the following two ways and obtain the first two families 
of optimal pure EAQLRCs with flexible parameters and nontrivial localities. 
\begin{itemize}
  \item We take $\mat H_X$ and $\mat H_Z$ to be parity-check matrices of 
  $\mathcal C_X=\mathcal C_1^\perp$ and $\mathcal C_Z=\mathcal C_2$, respectively, 
  where $\mathcal C_1$ and $\mathcal C_2$ are MDS codes satisfying 
  $\dim(\mathcal C_1\cap\mathcal C_2)=d-1-c$. 
  We set $\mat D=\mat I_n$. 
For $q\ge3$ and integers $n,d,c$ satisfying
$2\le d\le n\le q+1$ and
$\max\{0,2d-n-1\}\le c\le d-1$, except for the parameter triples
$(q+1,2,1)$ and $(q+1,q+1,q)$, as well as
$(q+1,q/2+1,0)$ when $q$ is even, we construct a family of optimal pure
$\dsb{n,n+c-2d+2,d;c}_q$ EAQLRCs with locality $n-d+1$ in
Theorem~\ref{thm.mds-eaqlrc-spectrum}. 
  We further show that this family can have nontrivial locality 
  even though its classical codes are MDS in Remark~\ref{rem.con1}.

  \item We take $\mat H_X=\mat H_Z=\mat H$, where $\mat H$ is the block parity-check matrix 
  in \eqref{eq:~grs-block-check222}, and choose the nonidentity diagonal matrix 
  $\mat D=\mat D_s$ in \eqref{eq.D_22}. 
  We then use $\mat H_Z'=\mat H_Z\mat D_s$ as the parity-check matrix of the second classical code. 
  Under the arithmetic conditions of Theorem~\ref{thm.grs-coset-flexible}, 
  we construct a family of optimal pure 
  $\dsb{u(r+1),ur-2d+4+s,d;u+s}_q$ EAQLRCs with locality $r$ 
  for every $0\le s\le d-2$. 
\end{itemize}
We also provide two explicit examples in Examples~\ref{exam.1} and~\ref{exam.2} 
  to illustrate the constructions of optimal pure EAQLRCs with nontrivial localities 
  in Theorems~\ref{thm.mds-eaqlrc-spectrum} and~\ref{thm.grs-coset-flexible}, respectively. 
\end{enumerate}

The remainder of this paper is organized as follows.
We review the necessary background on classical codes and quantum codes 
in Section~\ref{sec:preliminaries}.
We define EAQLRCs, characterize sufficient conditions under which EASCs have locality $r$, 
and present the CSS-like construction of EAQLRCs in Section~\ref{sec:model-css}.
We derive a locality upper bound and a Singleton-like bound 
in Section~\ref{sec:bounds} for arbitrary EAQLRCs arising from 
the CSS-like construction in Section~\ref{sec:model-css}. 
We also establish necessary and sufficient conditions for pure EAQLRCs 
to attain equality in the Singleton-like bound. 
In Section~\ref{sec:constructions}, we develop a general framework for constructing 
optimal pure CSS-like EAQLRCs and construct the first two families of optimal pure EAQLRCs 
with flexible parameters and nontrivial localities.  
Finally, we conclude the paper in Section~\ref{sec:~conclusion}.

\section{Preliminaries}\label{sec:preliminaries}

For any vector $\vect{u}=(u_1,u_2,\ldots,u_n)\in\F_q^n$, 
let $\supp(\vect{u})=\{i:~u_i\ne0\}$ be the {\em support} of $\vect{u}$, 
and let $\wt(\vect{u})=|\supp(\vect{u})|$ be the {\em (Hamming) weight} of $\vect{u}$.
For any subset $S\subseteq\F_q^n$ containing a nonzero vector, let $\wt(S)$ be
the {\em minimum weight} of all nonzero vectors in $S$. 
An $[n,k,d]_q$ linear code $\mathcal{C}$ is a $k$-dimensional subspace of
$\F_q^n$ with minimum distance $d$. 
We also write $\dim(\mathcal C)$ and $d(\mathcal C)$ for its {\em dimension} and {\em minimum distance}. 
For a matrix $\mat{H}$, we use $\ker(\mat{H})$ to denote its \emph{null space} and
$\row(\mat{H})$ to denote its \emph{row space}.  A \emph{parity-check matrix} $\mat{H}$ for
an $[n,k]_q$ linear code $\mathcal{C}$ is an $(n-k)\times n$ matrix chosen over
$\F_q$ such that
$\mathcal{C}=\ker(\mat{H})$ and $\mathcal{C}^\perp=\row(\mat{H})$.

We first recall several basic facts on classical locally recoverable codes
(cLRCs).  Following Gopalan $et~al.$ \cite{GopalanHuangSimitciYekhanin2012},
a classical LRC permits each erased coordinate of a codeword to be recovered
from a set of other coordinates in the same codeword.

\begin{definition}
An $[n,k,d]_q$ linear code $\mathcal{C}$ is
said to have \emph{locality} $r$ if
each symbol $c_i$ of a codeword $\vect{c}=(c_1,c_2,\ldots,c_n)\in\mathcal C$ with $i\in[n]$ 
can be recovered from at most $r$ other distinct symbols of $\vect{c}$.  
More precisely, for each $i\in[n]$, there exist a subset
$I_i=\{i_1,i_2,\ldots,i_{s_i}\}\subseteq[n]\setminus\{i\}$ with $s_i\le r$
and a function $f_i:\F_q^{s_i}\to\F_q$ such that
$c_i=f_i(c_{i_1},c_{i_2},\ldots,c_{i_{s_i}})$ for every $\vect c\in\mathcal C$.
The set $I_i$ is called a \emph{recovery set} for $c_i$.
\end{definition}

Given a linear code, we have the following characterization on its locality 
and the Singleton-like bound on its parameters. 

\begin{lemma}{\rm (\!\! \cite[Definition 1]{LuoChenEzermanLing2025})}\label{lem.classical-locality-dual}
An $[n,k,d]_q$ linear code $\mathcal{C}$ has locality $r$ if, for
each $i\in[n]$, there exists a codeword
$\vect{h}_i\in\mathcal C^\perp$ such that
$i\in\supp(\vect{h}_i)$ and $|\supp(\vect{h}_i)|\le r+1$.
\end{lemma}

\begin{lemma}{\rm (\!\! \cite[Theorem 5]{GopalanHuangSimitciYekhanin2012})}\label{lem.classical-singleton}
  For an $[n,k,d]_q$ linear code $\mathcal{C}$ with locality $r$, we have 
  \begin{equation}\label{eq.singleton-classical} d\le n-k-\left\lceil\frac{k}{r}\right\rceil+2. \end{equation}
\end{lemma}

Let $q=p^m$ be a prime power and let
${\rm tr}:\F_q\to\F_p$ be the {\em field trace}. Set
$\omega_p=e^{2\pi i/p}$. For $a,b,x\in\F_q$, define the {\em $q$-ary Pauli operators} by
$$X(a)\ket{x}=\ket{x+a}~{\rm and}~Z(b)\ket{x} = \omega_p^{{\rm tr}(bx)}\ket{x}.$$
They satisfy
$Z(b)X(a)=\omega_p^{{\rm tr}(ab)}X(a)Z(b)$ and are unitary, with
$X(a)^\dagger=X(-a)$ and $Z(b)^\dagger=Z(-b)$.
For an ordered register $A=(A_1,A_2,\ldots,A_t)$ of $t$ qudits and
$\vect a=(a_1,a_2,\ldots,a_t),\vect b=(b_1,b_2,\ldots,b_t)\in\F_q^t$, put
$$X_A(\vect a)=\bigotimes_{j=1}^tX(a_j)~{\rm and}~Z_A(\vect b)=\bigotimes_{j=1}^tZ(b_j).$$
These tensor-product operators are again unitary.
The $q$-ary commutation relation between
$X_A(\vect a)Z_A(\vect b)$ and
$X_A(\vect a')Z_A(\vect b')$ is
\begin{equation}\label{eq:pauli-commutation} X_A(\vect a)Z_A(\vect b) X_A(\vect a')Z_A(\vect b') = \omega_p^{ {\rm tr}\left( \langle\vect a',\vect b\rangle - \langle\vect a,\vect b'\rangle \right) } X_A(\vect a')Z_A(\vect b') X_A(\vect a)Z_A(\vect b). \end{equation}

The {\em $q$-ary Pauli group} on $\calH_A$, including scalar phases, is 
\begin{equation}\label{eq.error-group} \mathcal G_A(q) = \left\{ \omega_p^\ell X_A(\vect a)Z_A(\vect b):~\vect a,\vect b\in\F_q^t,\ \ell\in\mathbb Z_p \right\}. \end{equation}
By Definition~\ref{def.eaqecc}, every transmitted Pauli error occurring in an EAQECC has the form
$E_{Q_{[n]}}\otimes I_{B_{[c]}}$, where
$E_{Q_{[n]}}\in\mathcal G_{Q_{[n]}}(q)$ and $I_{B_{[c]}}$ is the
identity operator on Bob's register $B_{[c]}$.  The joint Pauli group
$\mathcal G_{Q_{[n]}\cup B_{[c]}}(q)$ is used below to define
entanglement-assisted stabilizer codes. 

\begin{definition}\label{def:ea-stabilizer-code}
An $\dsb{n,\kappa,\delta;c}_q$
\emph{entanglement-assisted stabilizer code (EASC)} is an EAQECC
$\mathcal Q$ for which there exists an abelian subgroup
$S_{\rm EA}\subseteq \mathcal G_{Q_{[n]}\cup B_{[c]}}(q)$ such that
\begin{equation}\label{eq:~easc-code-space} \mathcal Q = \left\{ \ket{\psi}\in \calH_{Q_{[n]}}\otimes\calH_{B_{[c]}}:~G\ket{\psi}=\ket{\psi} \text{ for every }G\in S_{\rm EA} \right\}. \end{equation}
\end{definition}

The group $S_{\rm EA}$ is called the \emph{extended stabilizer group},
and its elements are called \emph{extended stabilizers}. 
Unlike a transmitted channel error $E_{Q_{[n]}}\otimes I_{B_{[c]}}$, an extended stabilizer may act nontrivially on Bob's register. 
With \eqref{eq.error-group} and \eqref{eq:~easc-code-space}, every $G\in S_{\rm EA}$ can be written as 
$$G = \omega_p^\ell X_{Q_{[n]}}(\vect a)Z_{Q_{[n]}}(\vect b) \otimes X_{B_{[c]}}(\vect r)Z_{B_{[c]}}(\vect s),$$
where $\vect a,\vect b\in\F_q^n$, $\vect r,\vect s\in\F_q^c$, and
$\ell\in\mathbb \F_p$.
Its \emph{Pauli label} and {\em support} are, respectively,
$${\rm lab}(G) =(\vect a,\vect r\mid\vect b,\vect s) \in\F_q^{2(n+c)}~{\rm and}~\supp_Q(G)=\{j\in[n]:~(a_j,b_j)\ne(0,0)\}.$$
The \emph{Pauli label set} of $S_{\rm EA}$ is
$L(S_{\rm EA}) =\{{\rm lab}(G):~G\in S_{\rm EA}\} \subseteq\F_q^{2(n+c)}.$
Moreover, the identity
${\rm lab}(GG')={\rm lab}(G)+{\rm lab}(G')$ for
$G,G'\in S_{\rm EA}$ shows that $L(S_{\rm EA})$ is an $\F_p$-linear
subspace of $\F_q^{2(n+c)}$.  To prove label uniqueness, suppose that
$\lambda I\in S_{\rm EA}$.  For any nonzero
$\ket{\psi}\in\mathcal Q$,
$\lambda\ket{\psi}=(\lambda I)\ket{\psi}=\ket{\psi}$, which implies
$\lambda=1$.  If $G,G'\in S_{\rm EA}$ have the same Pauli label, then
$G'G^{-1}=\lambda I\in S_{\rm EA}$ for some scalar $\lambda$.  The preceding
argument gives $\lambda=1$, so $G'=G$.  Thus, every element of $L(S_{\rm EA})$
determines a unique extended stabilizer despite the omitted scalar phase.

We recall the EA-CSS construction from pairs of classical linear codes.  The
resulting EAQECCs are EASCs.

\begin{lemma}{\rm (\!\! EA-CSS construction \cite[Theorems~2 and~4]{GalindoHernandoMatsumotoRuano2019})}\label{lem:CSS-construction}
Let $\mathcal{C}_X$ and $\mathcal{C}_Z$ be $[n,k_X,d_X]_q$ and $[n,k_Z,d_Z]_q$
linear codes, respectively.  Then there exists an EAQECC $\mathcal{Q}(\C_X,\C_Z)$ with parameters
$\dsb{n,\kappa,\delta;c}$, where
\begin{equation}\label{eq.eaqecc-css}
  \begin{aligned}
  c &= n-k_X-\dim(\mathcal C_X^\perp\cap\mathcal C_Z) =n-k_Z-\dim(\mathcal C_X\cap\mathcal C_Z^\perp),\\
  \kappa &=k_X+k_Z-n+c~{\rm and}~\\
  \delta &= \left\{\begin{array}{ll} \min\{d_X,d_Z\}, & \text{if }\mathcal C_X \subseteq \mathcal C_Z^\perp,\\[1ex]
  \min\{ \wt(\mathcal C_X\setminus(\mathcal C_X\cap\mathcal C_Z^\perp)), \wt(\mathcal C_Z\setminus(\mathcal C_X^\perp\cap\mathcal C_Z))\}, & \text{otherwise.}
  \end{array}
  \right.
  \end{aligned}
\end{equation}
Moreover, the EAQECC $\mathcal{Q}(\C_X,\C_Z)$ is called \emph{pure} if
$\delta=\min\{d_X,d_Z\}$ and \emph{impure} otherwise.
\end{lemma}

\begin{lemma}{\rm (\!\! \cite[Proposition~4]{GalindoHernandoMatsumotoRuano2019})}\label{lem.defect}
Let $\mathcal{C}_X$ and $\mathcal{C}_Z$ be $[n,k_X]_q$ and $[n,k_Z]_q$ linear
codes, respectively.  
If $\mat{H}_X$ and $\mat{H}_Z$ are respective parity-check matrices of $\mathcal{C}_X$ and $\mathcal{C}_Z$,
then 
\begin{equation*}
  \begin{aligned}
  \dim(\mathcal C_X^\perp\cap\mathcal C_Z) & = n-k_X-\rank(\mat{H}_X\mat{H}_Z^{\top})~{\rm and}~\\
  \dim(\mathcal C_X\cap\mathcal C_Z^\perp) & = n-k_Z-\rank(\mat{H}_X\mat{H}_Z^{\top}).
  \end{aligned}
\end{equation*}
\end{lemma}

For the $\dsb{n,\kappa,\delta;c}_q$ EA-CSS code $\mathcal Q(\C_X,\C_Z)$ obtained from
Lemma~\ref{lem:CSS-construction}, let $\mat H_X$ and $\mat H_Z$ be parity-check
matrices of $\mathcal C_X$ and $\mathcal C_Z$, respectively.  By
Lemma~\ref{lem.defect}, $c=\rank(\mat H_X\mat H_Z^{\top})$.  
Choose a rank factorization
$$-\mat H_X\mat H_Z^{\top}=\mat F_X\mat F_Z^{\top}~{\rm with}~
\mat F_X\in\F_q^{(n-k_X)\times c}~{\rm and}~
\mat F_Z\in\F_q^{(n-k_Z)\times c}.$$ 
Then the extended stabilizer group of the EA-CSS code $\mathcal{Q}(\C_X,\C_Z)$ is
\begin{equation}\label{eq:~ea-css-extended-stabilizer}
  S_{\rm EA}=\{X_Q(\vect u\mat H_X)Z_Q(\vect v\mat H_Z)\otimes X_B(\vect u\mat F_X)Z_B(\vect v\mat F_Z):~
  \vect u\in\F_q^{n-k_X},\ \vect v\in\F_q^{n-k_Z}\}. 
\end{equation}
In fact, with the notation $G(\vect u,\vect v)=X_Q(\vect u\mat H_X)Z_Q(\vect v\mat H_Z)\otimes X_B(\vect u\mat F_X)Z_B(\vect v\mat F_Z)$, 
it follows from \eqref{eq:pauli-commutation} that 
$$G(\vect u,\vect v)G(\vect u',\vect v')
=\omega_p^{{\rm tr}\left(
\vect u'(\mat H_X\mat H_Z^{\top}
+\mat F_X\mat F_Z^{\top})\vect v^{\top}\right)}
G(\vect u+\vect u',\vect v+\vect v')
=G(\vect u+\vect u',\vect v+\vect v').$$
Thus, the set $S_{\rm EA}$ in \eqref{eq:~ea-css-extended-stabilizer} is closed under
multiplication without additional scalar phases. In particular, every
element is represented with scalar phase $1$, and $L(S_{\rm EA})$ is
$\F_q$-linear.

\begin{remark}
If $c=0$, then $\mat F_X$ and $\mat F_Z$ have no columns, 
implying that 
$
X_B(\vect u\mat F_X)Z_B(\vect v\mat F_Z)=I_{B_{[0]}}.
$
In this case, the extended stabilizer group $S_{\rm EA}$ reduces to
\begin{align*}
  \begin{split}
    S_{\rm CSS} & = \{X_Q(\vect u\mat H_X)Z_Q(\vect v\mat H_Z)\otimes I_{B_{[0]}}:~\vect u\in\F_q^{n-k_X},\ \vect v\in\F_q^{n-k_Z}\} \\
                & = \{X_Q(\vect h_X)Z_Q(\vect h_Z):~\vect h_X\in\mathcal C_X^\perp,\ \vect h_Z\in\mathcal C_Z^\perp\},
  \end{split}
\end{align*}
which is exactly the stabilizer group of an ordinary quantum CSS code (see for example \cite[Definition~10]{LiJinXing2026Asymmetric}).
\end{remark}

\section{Characterizations and construction methods of EAQLRCs}\label{sec:model-css}

In this section, we define locality for EAQECCs and give sufficient
criteria for EASCs to have locality $r$.

\subsection{Introducing locality into EAQECCs}

According to Definition~\ref{def.eaqecc}, the transmitted register $Q_{[n]}$
passes through the noisy channel, whereas Bob's register $B_{[c]}$ is retained
and remains reliable.  Therefore, the definition of locality must count
access to transmitted qudits separately from access to Bob's qudits. 

For a finite register $A$, let $\calM(A)$ denote the {\em set of density operators} 
on $\calH_A$, namely, the positive semidefinite operators $\rho$ on $\calH_A$ satisfying $\Tr(\rho)=1$.
Given an EAQECC  $\mathcal Q\subseteq\calH_{Q_{[n]}}\otimes\calH_{B_{[c]}}$, 
a density operator $\rho\in\calM(Q_{[n]}\cup B_{[c]})$ is called a \emph{code state} 
of $\mathcal Q$ if it can be written as
$\rho=\sum_j p_j\ket{\psi_j}\bra{\psi_j},$ 
where each $\ket{\psi_j}$ is a unit vector in $\mathcal Q$, 
$p_j\ge0$, and $\sum_jp_j=1$.
A code state is called \emph{pure} if $\rho=\ket{\psi}\bra{\psi}$ for some unit vector $\ket{\psi}\in\mathcal Q$.
By the usual identification of a pure state with its rank-one density operator, 
we also refer to $\ket{\psi}$ as a pure code state when no confusion can arise. 
We write $I_A$ for the identity operator on $\calH_A$, $\id_A$ for the identity channel on $A$, and
$U^\dagger$ for the conjugate transpose of a linear operator $U$.
For disjoint registers $A$ and $C$, we write $\Tr_A$ for the {\em partial trace} over $A$. 
For any orthonormal basis $\{\ket{x}\}$ of $\calH_A$ and every density
operator $\rho\in\calM(A\cup C)$, it is defined by
\begin{align}\label{eq:~partial-trace}
    \Tr_A(\rho)=\sum_x(\bra{x}\otimes I_C)\rho(\ket{x}\otimes I_C).
\end{align}

A {\em quantum channel} from register $A$ to register $C$ is a map
$\mathcal K:\calM(A)\to\calM(C)$ that admits a completely positive 
trace-preserving (CPTP) linear extension from the linear operators on
$\calH_A$ to those on $\calH_C$.
In finite dimensions, this is equivalent to the existence of linear
operators $K_\nu:\calH_A\to\calH_C$, called {\em Kraus operators}, such that
\begin{align}\label{eq.quantum-channel}
\mathcal K(\rho)=\sum_\nu K_\nu\rho K_\nu^\dagger
~{\rm and}~
\sum_\nu K_\nu^\dagger K_\nu=I_A.
\end{align}
In the studies of qLRCs \cite{GolowichGuruswami2023,GalindoHernandoMartinCruzMatsumoto2024,
LiJinXing2026Asymmetric,GuruswamiKshirsagarTrivedi2026},
the erasure of a qudit at a known position is modeled by the completely
depolarizing channel on that qudit.
Following the model in
\cite[Definition~14]{LiJinXing2026Asymmetric}, we apply the completely
depolarizing channel to one transmitted qudit and the identity channel
to all the other transmitted qudits and Bob's register. 
Then we have the following definition.

\begin{definition}\label{def.single-qudit-erasure-channel}
For $i\in[n]$, define the map 
$$\Gamma_i:\calM(Q_{[n]}\cup B_{[c]}) \longrightarrow \calM(Q_{[n]}\cup B_{[c]})$$
for every quantum state
$\rho\in\calM(Q_{[n]}\cup B_{[c]})$ by
$$\Gamma_i(\rho) = \frac{I_{Q_i}}{q}\otimes\Tr_{Q_i}(\rho).$$
Thus, $\Gamma_i$ traces out $Q_i$ and replaces it by a maximally mixed qudit,
while leaving all other labelled subsystems unchanged.
\end{definition}

For $i\in[n]$, let $X_i(a)$ and $Z_i(b)$ act as $X(a)$ and $Z(b)$ on
$Q_i$ and as the identity on $Q_{[n]\setminus\{i\}}\cup B_{[c]}$. 
The following identity \eqref{eq:~local-erasure-channel} is included in Lemma \ref{lem.single-qudit-erasure-channel} for completeness.  
It expresses $\Gamma_i$ in terms of the operators $X_i(a)Z_i(b)$ used below and verifies
that $\Gamma_i$ is a quantum channel.

\begin{lemma}\label{lem.single-qudit-erasure-channel}
For every $i\in[n]$ and $\rho\in\calM(Q_{[n]}\cup B_{[c]})$, the map $\Gamma_i$ in
Definition~\ref{def.single-qudit-erasure-channel} satisfies
\begin{equation}\label{eq:~local-erasure-channel} \Gamma_i(\rho)=\frac{1}{q^2}\sum_{a,b\in\F_q}\bigl(X_i(a)Z_i(b)\bigr)\rho\bigl(X_i(a)Z_i(b)\bigr)^\dagger. \end{equation}
Consequently, $\Gamma_i$ is a quantum channel.
\end{lemma}
\begin{proof}
Let $A=Q_{[n]\setminus\{i\}}\cup B_{[c]}$. 
Since $\{\ket{x}\bra{y}:~x,y\in\F_q\}$ is a basis for the linear operators on
$\calH_{Q_i}$, the density operator $\rho$ has a unique block decomposition
$\rho=\sum_{x,y\in\F_q}\ket{x}\bra{y}\otimes\rho_{x,y},$
where each $\rho_{x,y}$ is a linear operator on $\calH_A$ given by
$\rho_{x,y}=(\bra{x}\otimes I_A)\rho(\ket{y}\otimes I_A).$
With \eqref{eq:~partial-trace}, we derive  
$\Tr_{Q_i}(\rho)=\sum_{x\in\F_q}\rho_{x,x}.$
Note that $\sum_{b\in\F_q}\omega_p^{{\rm tr}(bt)}=0$ for $t\in\F_q^*$ 
and $\sum_{b\in\F_q}\omega_p^{{\rm tr}(bt)}=q$ for $t=0$.  
Together with the facts $X(a)Z(b)\ket{x}=\omega_p^{{\rm tr}(bx)}\ket{x+a}$ and $\sum_{a\in\F_q}\ket{x+a}\bra{x+a}=I_{Q_i}$, 
since $X_i(a)Z_i(b)$ acts as $X(a)Z(b)$ on $Q_i$ and as the identity on $A$, it can be verified that 
$$
\begin{aligned}
\frac{1}{q^2}\sum_{a,b\in\F_q}\bigl(X_i(a)Z_i(b)\bigr)\rho\bigl(X_i(a)Z_i(b)\bigr)^\dagger
= & \frac{1}{q^2}\sum_{a,b,x,y\in\F_q}\omega_p^{{\rm tr}(b(x-y))}
\ket{x+a}\bra{y+a}\otimes\rho_{x,y}\\
= & \frac{1}{q}\sum_{a,x\in\F_q}\ket{x+a}\bra{x+a}\otimes\rho_{x,x}\\
= & \frac{I_{Q_i}}{q}\otimes\sum_{x\in\F_q}\rho_{x,x} \\
= & \frac{I_{Q_i}}{q}\otimes\Tr_{Q_i}(\rho)\\
= & \Gamma_i(\rho).
\end{aligned}
$$
This proves \eqref{eq:~local-erasure-channel}.

It remains to verify that $\Gamma_i$ is a quantum channel. 
With \eqref{eq:~local-erasure-channel}, it is clear to see that $\Gamma_i$ is a linear map. 
Hence, it suffices to find a family of linear operators satisfying the two
conditions in \eqref{eq.quantum-channel}. 
Take $K_{a,b}=\frac{X_i(a)Z_i(b)}{q}$ for $(a,b)\in\F_q^2$. 
Then the first condition in \eqref{eq.quantum-channel} immediately holds from \eqref{eq:~local-erasure-channel}.
Since every $X_i(a)Z_i(b)$ is unitary,
$$\sum_{a,b\in\F_q}K_{a,b}^\dagger K_{a,b}
=\frac{1}{q^2}\sum_{a,b\in\F_q}I_{Q_{[n]}\cup B_{[c]}}
=I_{Q_{[n]}\cup B_{[c]}}.$$
This completes the proof. 
\end{proof}

Recall that a code state of an EAQECC can be pure or mixed, 
and every mixed code state is a convex combination of pure code states.  
Since quantum channels are linear, recovery of every pure code state also recovers every mixed code state in an EAQECC.  
These facts yield the following definitions of recovery sets and localities for EAQECCs.

\begin{definition}\label{def.ij-local}
Let $\mathcal Q\subseteq \calH_{Q_{[n]}}\otimes\calH_{B_{[c]}}$ be an EAQECC.  
Let $i\in[n]$ and let $R\subseteq[n]\setminus\{i\}$.  The set $R$ is a \emph{recovery set} for $Q_i$
if there exists a quantum channel
$$\Rec_{i,R}:\calM(Q_{\{i\}\cup R}\cup B_{[c]}) \longrightarrow \calM(Q_{\{i\}\cup R}\cup B_{[c]})$$
such that, for every pure code state $\ket{\psi}\in\mathcal Q$, one has
\begin{equation}\label{eq:~ij-recovery} (\Rec_{i,R}\otimes\id_{Q_{[n]\setminus(\{i\}\cup R)}}) \bigl(\Gamma_i(\ket{\psi}\bra{\psi})\bigr) =\ket{\psi}\bra{\psi}. \end{equation}
\end{definition}

\begin{definition}\label{def.eaqlrc}
An $\dsb{n,\kappa,\delta;c}_q$ EAQECC
$\mathcal{Q}$ is said to have {\em locality $r$} if, for every $i\in[n]$,
there exists $R_i\subseteq[n]\setminus\{i\}$ with $|R_i|\le r$ such that $R_i$
is a recovery set for $Q_i$.
We also call $\mathcal{Q}$ an $\dsb{n,\kappa,\delta;c}_q$ {\em entanglement-assisted quantum locally recoverable code (EAQLRC) 
with locality $r$}.  
\end{definition}

\subsection{A characterization of EASCs with locality}

In the following theorem, we characterize sufficient conditions for an EASC to have 
locality $r$ in terms of properties of its extended stabilizers.  
The proof constructs a local recovery channel by distinguishing all Pauli
errors on an erased transmitted qudit through their commutation relations
with suitable extended stabilizers.

\begin{theorem}\label{thm.stabilizer-locality}
Let $\mathcal Q\subseteq \mathcal H_{Q_{[n]}}\otimes \mathcal H_{B_{[c]}}$
be an EASC such that $L(S_{\rm EA})$ is $\F_q$-linear.  Suppose that for
every $i\in[n]$, there exist $G_i,K_i\in S_{\rm EA}$ whose respective Pauli
labels are
$${\rm lab}(G_i) = (\vect a_i,\vect r_i\mid\vect b_i,\vect s_i)~{\rm and}~{\rm lab}(K_i) = (\vect u_i,\vect t_i\mid\vect v_i,\vect w_i),$$
where
$\vect a_i,\vect b_i,\vect u_i,\vect v_i\in\F_q^n$ and
$\vect r_i,\vect s_i,\vect t_i,\vect w_i\in\F_q^c$, and such that
\begin{enumerate}
  \item [\rm 1)]
  $\bigl((\vect a_i)_i,(\vect b_i)_i\bigr)=(1,0)$ and
  $\bigl((\vect u_i)_i,(\vect v_i)_i\bigr)=(0,1)$.

  \item [\rm 2)]
  $\left|
 (\supp_Q(G_i)\cup\supp_Q(K_i))
  \setminus\{i\}
  \right|\le r$.
\end{enumerate}
Then $\mathcal Q$ is an EAQLRC with locality $r$.
\end{theorem}
\begin{proof}
For a given index $i\in[n]$, let
$R_i=(\supp_Q(G_i)\cup\supp_Q(K_i))\setminus\{i\}$.
Then $|R_i|\le r$ by condition {\rm 2)}.
Let $q=p^m$ and let $\eta_1,\eta_2,\ldots,\eta_m$ be an
$\F_p$-basis of $\F_q$.  As shown after
Definition~\ref{def:ea-stabilizer-code}, the label map is injective on
$S_{\rm EA}$.
Since $L(S_{\rm EA})$ is $\F_q$-linear, for each $j\in[m]$ there are unique
$G_{i,j},K_{i,j}\in S_{\rm EA}$ satisfying
$${\rm lab}(G_{i,j}) = \eta_j(\vect a_i,\vect r_i\mid\vect b_i,\vect s_i)~{\rm and}~{\rm lab}(K_{i,j}) = \eta_j(\vect u_i,\vect t_i\mid\vect v_i,\vect w_i).$$
Let
$\mathcal S_i=\{G_{i,j},K_{i,j}:~1\le j\le m\}$.
Their Pauli labels at $Q_i$ are respectively $(\eta_j,0)$ and $(0,\eta_j)$,
implying that the $2m$ operators in $\mathcal S_i$ are distinct.  Moreover, every
operator in $\mathcal S_i$ acts trivially on the transmitted qudits outside
$\{i\}\cup R_i$.

Define
$\mathcal E_i=\{E_{x,z}=X_i(x)Z_i(z):~x,z\in\F_q\}$, where every
$E_{x,z}$ acts as the identity on
$Q_{[n]\setminus\{i\}}\cup B_{[c]}$.  For Pauli operators $A$ and $E$, define
$\sigma(A,E)\in\F_p$ by
$AE=\omega_p^{\sigma(A,E)}EA$.
Note also that the exponent $\sigma(A,E)$ is unchanged if either operator is multiplied by a scalar
phase and hence depends only on their Pauli labels.
Order $\mathcal S_i$ as
$(G_{i,1},G_{i,2},\ldots,G_{i,m},
K_{i,1},K_{i,2},\ldots,K_{i,m})$ and define
$$
\begin{aligned}
\sigma_{\mathcal S_i}(E)
=\bigl(\sigma(G_{i,1},E),\sigma(G_{i,2},E),\ldots,\sigma(G_{i,m},E),
\sigma(K_{i,1},E),\sigma(K_{i,2},E),\ldots,\sigma(K_{i,m},E)\bigr)
\in\F_p^{2m}.
\end{aligned}
$$
By the $q$-ary Pauli commutation relation in
\eqref{eq:pauli-commutation} and condition {\rm 1)}, we have
$\sigma(G_{i,j},E_{x,z})=-{\rm tr}(\eta_jz)$ and
$\sigma(K_{i,j},E_{x,z})={\rm tr}(\eta_jx)$, where
${\rm tr}$ is the field trace from $\F_q$ to $\F_p$.  
It then follows that
$$\sigma_{\mathcal S_i}(E_{x,z}) = \bigl( -{\rm tr}(\eta_1z),-{\rm tr}(\eta_2z), \ldots,-{\rm tr}(\eta_mz), {\rm tr}(\eta_1x),{\rm tr}(\eta_2x), \ldots,{\rm tr}(\eta_mx) \bigr).$$ 
If
$\sigma_{\mathcal S_i}(E_{x,z})
=\sigma_{\mathcal S_i}(E_{x',z'})$,
then
${\rm tr}(\eta_j(z-z'))=0$ and
${\rm tr}(\eta_j(x-x'))=0$ for every $j\in[m]$.
Since $\{\eta_1,\eta_2,\ldots,\eta_m\}$ is an $\F_p$-basis and the field trace is
$\F_p$-linear, we deduce that 
${\rm tr}(\lambda(z-z'))=0$ and
${\rm tr}(\lambda(x-x'))=0$ for every $\lambda\in\F_q$. 
In addition, ${\rm tr}(\lambda\mu)=0$ for every $\lambda\in\F_q$ only if $\mu=0$, 
as the bilinear form $(\lambda,\mu)\mapsto{\rm tr}(\lambda\mu)$ is nondegenerate. 
Hence, $z=z'$ and $x=x'$, that is to say the map $E\mapsto\sigma_{\mathcal S_i}(E)$ is injective on $\mathcal E_i$.
Since $|\mathcal E_i|=q^2=p^{2m}=|\F_p^{2m}|$, this map is bijective.

By Definition~\ref{def:ea-stabilizer-code}, the operators in
$\mathcal S_i$ commute and fix $\mathcal Q$ pointwise. Since they are
Pauli operators, they are unitary. Moreover, for every
$G\in\mathcal S_i$, we have $G^p\in S_{\rm EA}$ and
$
{\rm lab}(G^p)
=p{\rm lab}(G)
=0
={\rm lab}(I).
$
The injectivity of the label map on $S_{\rm EA}$ therefore gives
$G^p=I$.
Each $G\in\mathcal S_i$ acts trivially on
$Q_{[n]\setminus(\{i\}\cup R_i)}$. Hence, including its scalar phase,
$G$ is the identity extension of an operator on
$\calH_{Q_{\{i\}\cup R_i}\cup B_{[c]}}$.
We henceforth use the same symbol $G$ for this local operator; whenever
it acts on the full system, its identity action on
$Q_{[n]\setminus(\{i\}\cup R_i)}$ is understood. 
For $\mathbf s=(s_G)_{G\in\mathcal S_i}\in\F_p^{2m}$, let
$\Pi_{\mathbf s}$ be the orthogonal projector onto
$$
\left\{
\ket{\phi}\in\calH_{Q_{\{i\}\cup R_i}\cup B_{[c]}}
:~
G\ket{\phi}=\omega_p^{s_G}\ket{\phi}
\text{ for every }G\in\mathcal S_i
\right\}.
$$
Since the operators in $\mathcal S_i$ commute and satisfy $G^p=I$,
this projector has the explicit expression
$
\Pi_{\mathbf s}
=
\prod_{G\in\mathcal S_i}
\left(
\frac1p\sum_{\ell=0}^{p-1}
\omega_p^{-\ell s_G}G^\ell
\right).
$
The nonzero joint eigenspaces of the commuting unitaries in
$\mathcal S_i$ form an orthogonal decomposition of
$\calH_{Q_{\{i\}\cup R_i}\cup B_{[c]}}$. 
This yields that 
$$
\sum_{\mathbf s\in\F_p^{2m}}\Pi_{\mathbf s}
=
I_{Q_{\{i\}\cup R_i}\cup B_{[c]}}.
$$

We now use the projectors $\Pi_{\mathbf s}$ to construct a local
recovery channel. Since the map
$E\mapsto\sigma_{\mathcal S_i}(E)$ is bijective, for every
$\mathbf s\in\F_p^{2m}$ there exists a unique
$E_{\mathbf s}\in\mathcal E_i$ such that
$\sigma_{\mathcal S_i}(E_{\mathbf s})=\mathbf s$.
Let $C_{\mathbf s}=E_{\mathbf s}^\dagger$, regarded as an operator on
$\calH_{Q_{\{i\}\cup R_i}\cup B_{[c]}}$.
It acts nontrivially only on $Q_i$ and as the identity on
$Q_{R_i}\cup B_{[c]}$.
For $\tau\in \calM(Q_{\{i\}\cup R_i}\cup B_{[c]})$, define 
\begin{align}\label{eq.recover-channel-EASC}
  \Rec_{i,R_i}(\tau)
=
\sum_{\mathbf s\in\F_p^{2m}}
C_{\mathbf s}\Pi_{\mathbf s}\tau
\Pi_{\mathbf s}C_{\mathbf s}^\dagger,
\end{align}
which is clearly linear in $\tau$. 
Since $E_{\mathbf s}$ is a Pauli operator,
$C_{\mathbf s}$ is unitary and satisfies
$
C_{\mathbf s}^\dagger C_{\mathbf s}
=
I_{Q_{\{i\}\cup R_i}\cup B_{[c]}}.
$
Moreover, since $\Pi_{\mathbf s}$ is an orthogonal projector,
$\Pi_{\mathbf s}^\dagger=\Pi_{\mathbf s}^2=\Pi_{\mathbf s}$.
It follows that
$$
\sum_{\mathbf s\in\F_p^{2m}}
(C_{\mathbf s}\Pi_{\mathbf s})^\dagger
(C_{\mathbf s}\Pi_{\mathbf s})
=
\sum_{\mathbf s\in\F_p^{2m}}
\Pi_{\mathbf s}C_{\mathbf s}^\dagger C_{\mathbf s}\Pi_{\mathbf s}
=
\sum_{\mathbf s\in\F_p^{2m}}\Pi_{\mathbf s}
=
I_{Q_{\{i\}\cup R_i}\cup B_{[c]}}.
$$
Thus, $\Rec_{i,R_i}$ has Kraus operators $\{C_{\mathbf s}\Pi_{\mathbf s}:~\mathbf s\in\F_p^{2m}\}$ 
satisfying the conditions in \eqref{eq.quantum-channel}, and hence is a quantum channel.

Fix a unit vector $\ket{\psi}\in\mathcal Q$ and let $\rho_\psi=\ket{\psi}\bra{\psi}$.
Since every operator in $\mathcal S_i$ fixes $\ket{\psi}$, for every
$E\in\mathcal E_i$ and $G\in\mathcal S_i$, we have
$$
GE\ket{\psi}
=
\omega_p^{\sigma(G,E)}EG\ket{\psi}
=
\omega_p^{\sigma(G,E)}E\ket{\psi}.
$$
It implies that the corrupted state $E\ket{\psi}$ belongs to the joint eigenspace indexed by
$\sigma_{\mathcal S_i}(E)$. Moreover, it can be checked that 
$$
\bigl(
\Pi_{\mathbf s}\otimes
I_{Q_{[n]\setminus(\{i\}\cup R_i)}}
\bigr)E\ket{\psi}
=
\begin{cases}
E\ket{\psi},
&\mathbf s=\sigma_{\mathcal S_i}(E),\\
0,
&\mathbf s\ne\sigma_{\mathcal S_i}(E) 
\end{cases}
$$
holds for  every $\mathbf s\in\F_p^{2m}$. 
By the bijectivity of $E\mapsto\sigma_{\mathcal S_i}(E)$, we also derive that 
$E_{\sigma_{\mathcal S_i}(E)}=E$ and
$C_{\sigma_{\mathcal S_i}(E)}=E^\dagger$.
Hence, only the term indexed by
$\mathbf s=\sigma_{\mathcal S_i}(E)$ remains in the definition of
$\Rec_{i,R_i}$ in \eqref{eq.recover-channel-EASC}. It then follows that
$$
\begin{aligned}
&\bigl(
\Rec_{i,R_i}\otimes
\id_{Q_{[n]\setminus(\{i\}\cup R_i)}}
\bigr)
(E\rho_\psi E^\dagger)=E^\dagger E\rho_\psi E^\dagger E=\rho_\psi
\end{aligned}
$$
for every $E\in\mathcal E_i$.
Using \eqref{eq:~local-erasure-channel} and the linearity of
$\Rec_{i,R_i}$, we immediately obtain
$$
\begin{aligned}
\bigl(
\Rec_{i,R_i}\otimes
\id_{Q_{[n]\setminus(\{i\}\cup R_i)}}
\bigr)
\bigl(\Gamma_i(\rho_\psi)\bigr)
=  
\frac1{q^2}
\sum_{E\in\mathcal E_i}
\bigl(
\Rec_{i,R_i}\otimes
\id_{Q_{[n]\setminus(\{i\}\cup R_i)}}
\bigr)
(E\rho_\psi E^\dagger)
=  
\frac1{q^2}
\sum_{E\in\mathcal E_i}\rho_\psi
= 
\rho_\psi,
\end{aligned}
$$
where the last equality follows from $|\mathcal E_i|=q^2$.
This is precisely the recovery condition
\eqref{eq:~ij-recovery} with $R=R_i$. Hence, $R_i$ is a recovery set
for $Q_i$.
Since $i$ is arbitrary and $|R_i|\le r$, 
it turns out from Definition~\ref{def.eaqlrc} that $\mathcal Q$ is an EAQLRC with locality $r$.
This completes the proof. 
\end{proof}

\subsection{CSS-like constructions of EAQLRCs}

Recall that the label set of the extended stabilizer group $S_{\rm EA}$ in \eqref{eq:~ea-css-extended-stabilizer} 
of an EA-CSS code must be $\F_q$-linear.  
In this case, the extended stabilizers $G_i,~K_i\in S_{\rm EA}$ used in Theorem~\ref{thm.stabilizer-locality} 
can be characterized by dual codewords of the underlying classical codes $\mathcal C_X$ and $\mathcal C_Z$.

\begin{theorem}\label{thm.eacss}
Let $\mathcal{C}_X$ and $\mathcal{C}_Z$ be $[n,k_X,d_X]_q$ and
$[n,k_Z,d_Z]_q$ linear codes with parity-check matrices
$\mat{H}_X$ and $\mat{H}_Z$, respectively.
Assume that, for every $i\in[n]$, there exist
$\vect{h}_i^X\in\mathcal C_X^\perp$ and
$\vect{h}_i^Z\in\mathcal C_Z^\perp$ such that 
\begin{align}\label{eq.condition}
  i\in\supp(\vect{h}_i^X)\cap\supp(\vect{h}_i^Z) ~{\rm and}~ |(\supp(\vect{h}_i^X)\cup\supp(\vect{h}_i^Z))|\le r+1.
\end{align}
Then the EA-CSS code
$\mathcal{Q}(\mathcal{C}_X,\mathcal{C}_Z)$ is an
$\dsb{n,\kappa,\delta;c}_q$ EAQLRC of locality $r$, 
where the specific parameters are given in \eqref{eq.eaqecc-css}. 
\end{theorem}
\begin{proof}
By Lemma~\ref{lem:CSS-construction}, the pair
$(\mathcal{C}_X,\mathcal{C}_Z)$ gives an EA-CSS code
$\mathcal{Q}(\mathcal{C}_X,\mathcal{C}_Z)$ with parameters as in
\eqref{eq.eaqecc-css} and extended stabilizer group $S_{\rm EA}$ as in
\eqref{eq:~ea-css-extended-stabilizer}. 
Moreover, $S_{\rm EA}$ has an $\F_q$-linear Pauli label set. 
To verify locality, it remains to find extended stabilizers $G_i,K_i\in S_{\rm EA}$ 
satisfying the conditions 1) and 2) of Theorem~\ref{thm.stabilizer-locality} for each $i\in[n]$. 

For a given $i\in[n]$, by assumption, we can choose 
$\vect{h}_i^X\in\mathcal{C}_X^\perp$ and
$\vect{h}_i^Z\in\mathcal{C}_Z^\perp$ such that
$i\in\supp(\vect{h}_i^X)\cap\supp(\vect{h}_i^Z)$ and
$|(\supp(\vect{h}_i^X)\cup\supp(\vect{h}_i^Z))|\le r+1$.
Multiplying $\vect{h}_i^X$ and $\vect{h}_i^Z$ by nonzero scalars if
necessary, we may assume that
$(\vect{h}_i^X)_i=(\vect{h}_i^Z)_i=1$.
Since $\mat H_X$ and $\mat H_Z$ are parity-check matrices, there exist
coefficient vectors $\vect{y}_i^X\in\F_q^{n-k_X}$ and
$\vect{y}_i^Z\in\F_q^{n-k_Z}$ such that
$
\vect{h}_i^X=\vect{y}_i^X\mat H_X
~{\rm and}~
\vect{h}_i^Z=\vect{y}_i^Z\mat H_Z.
$
Taking
$(\vect{u},\vect{v})
=(\vect{y}_i^X,\vect{0}_{n-k_Z})$ and
$(\vect{u},\vect{v})
=(\vect{0}_{n-k_X},\vect{y}_i^Z)$, respectively, in
\eqref{eq:~ea-css-extended-stabilizer} gives the extended stabilizers
$$
G_i=
X_Q(\vect{h}_i^X)Z_Q(\vect{0}_n)
\otimes
X_B(\vect{y}_i^X\mat F_X)Z_B(\vect{0}_c)
~{\rm and}~
K_i=
X_Q(\vect{0}_n)Z_Q(\vect{h}_i^Z)
\otimes
X_B(\vect{0}_c)Z_B(\vect{y}_i^Z\mat F_Z).
$$
Their respective Pauli labels are
$
{\rm lab}(G_i)
=
(\vect{h}_i^X,\vect{y}_i^X\mat F_X
\mid\vect{0}_n,\vect{0}_c)
~{\rm and}~
{\rm lab}(K_i)
=
(\vect{0}_n,\vect{0}_c
\mid\vect{h}_i^Z,\vect{y}_i^Z\mat F_Z).
$
Therefore, it is easy to check that 
$$
\bigl((\vect{h}_i^X)_i,({\bf 0}_n)_i\bigr)=(1,0)
~{\rm and}~
\bigl(({\bf 0}_n)_i,(\vect{h}_i^Z)_i\bigr)=(0,1),
$$
and 
$$
\supp_Q(G_i)=\supp(\vect{h}_i^X)
~{\rm and}~
\supp_Q(K_i)=\supp(\vect{h}_i^Z).
$$
Since
$i\in\supp(\vect{h}_i^X)\cap\supp(\vect{h}_i^Z)$, we obtain
$$
\left|
\bigl(\supp_Q(G_i)\cup\supp_Q(K_i)\bigr)\setminus\{i\}
\right|
=
\left|
\supp(\vect{h}_i^X)\cup\supp(\vect{h}_i^Z)
\right|-1
\le r.
$$
We conclude that both the conditions {\rm 1) and 2)} of Theorem~\ref{thm.stabilizer-locality} hold. 
By Theorem~\ref{thm.stabilizer-locality}, $\mathcal Q(\mathcal C_X,\mathcal C_Z)$ is an EAQLRC with locality $r$.
\end{proof}

\begin{corollary}\label{cor.symmetric-eacss}
Let $\mathcal C$ be an $[n,k,d]_q$ linear code with locality $r$ and $\ell$-dimensional Euclidean hull.
Then the EA-CSS code $\mathcal Q(\mathcal C,\mathcal C)$ is an
$\dsb{n,k-\ell,\delta;n-k-\ell}_q$ EAQLRC of locality $r$, where
\begin{equation*}
\delta=
\begin{cases}
d, & \text{if } \mathcal C\subseteq\mathcal C^\perp,\\
\wt\bigl(\mathcal C\setminus(\mathcal C\cap\mathcal C^\perp)\bigr),
& \text{otherwise}.
\end{cases}
\end{equation*}
Moreover, the EAQLRC $\mathcal Q(\mathcal C,\mathcal C)$ is pure if $\delta=d$ and impure otherwise.
\end{corollary}
\begin{proof}

The desired result follows directly from combining with 
Lemma \ref{lem.classical-locality-dual} and Theorem~\ref{thm.eacss}.  
\end{proof}

\section{Two bounds for CSS-like EAQLRCs}\label{sec:bounds}

In this section, we derive two bounds for CSS-like EAQLRCs obtained from Theorem~\ref{thm.eacss}. 
The first is an upper bound on the locality of a CSS-like EAQLRC, 
which is independent of its distance and entanglement consumption.

\begin{theorem}\label{thm.elementary-locality}
Let $\mathcal C_X$ and $\mathcal C_Z$ be $[n,k_X,d_X]_q$ and
$[n,k_Z,d_Z]_q$ linear codes, respectively, where $d_X,d_Z\ge2$.
Let $r$ denote the minimum locality of the CSS-like EAQLRC
$\mathcal Q(\mathcal C_X,\mathcal C_Z)$ obtained from
Theorem~\ref{thm.eacss}. Then
$$r\le\min\{n-1,k_X+k_Z-\dim(\mathcal C_X\cap\mathcal C_Z)\}.$$
\end{theorem}
\begin{proof}
For any given $i\in[n]$, since $d_X,d_Z\ge2$, the punctured codes
$\C_X'$ and $\C_Z'$ obtained by deleting the $i^{\rm th}$ coordinate of
$\C_X$ and $\C_Z$ have dimensions $k_X$ and $k_Z$, respectively.
Let $t=\dim(\mathcal C_X\cap\mathcal C_Z)$. Since every nonzero codeword of
$\mathcal C_X\cap\mathcal C_Z$ has weight at least $2$, puncturing
$\mathcal C_X\cap\mathcal C_Z$ at the $i^{\rm th}$ coordinate preserves
its dimension. Hence, we can choose a set
$J_i\subseteq[n]\setminus\{i\}$ with $|J_i|=t$ such that the restriction of
$\mathcal C_X\cap\mathcal C_Z$ to $J_i$, denoted by
$(\mathcal C_X\cap\mathcal C_Z)_{J_i}$, has dimension $t$. 
Let $\mat G_X$ and $\mat G_Z$ be generator matrices of $\mathcal C_X$ and
$\mathcal C_Z$, respectively. Since $\mathcal C_X\cap\mathcal C_Z$ is a
subcode of both $\mathcal C_X$ and $\mathcal C_Z$ and
$\dim((\mathcal C_X\cap\mathcal C_Z)_{J_i})=t$, the columns indexed by
$J_i$ are linearly independent in both $\mat G_X$ and $\mat G_Z$.
Extend $J_i$ to column bases
$J_i\subseteq I_i^X\subseteq[n]\setminus\{i\}$ and
$J_i\subseteq I_i^Z\subseteq[n]\setminus\{i\}$ of generator matrices of
the punctured codes $\C_X'$ and $\C_Z'$, respectively. We then derive that
$|I_i^X|=k_X$, $|I_i^Z|=k_Z$, and
$|I_i^X\cup I_i^Z|\le k_X+k_Z-t$.

Since the $i^{\rm th}$ column of each generator matrix is a linear
combination of the columns in the corresponding basis, there exist
$\vect h_i^X\in\mathcal C_X^\perp$ and
$\vect h_i^Z\in\mathcal C_Z^\perp$ such that
$i\in\supp(\vect h_i^X)\cap\supp(\vect h_i^Z)$,
$\supp(\vect h_i^X)\subseteq I_i^X\cup\{i\}$, and
$\supp(\vect h_i^Z)\subseteq I_i^Z\cup\{i\}$.
It then follows that
$$|\supp(\vect h_i^X)\cup\supp(\vect h_i^Z)|\le 1+|I_i^X\cup I_i^Z|\le 1+k_X+k_Z-t=1+\dim(\mathcal C_X+\mathcal C_Z).$$
Moreover,
$\supp(\vect h_i^X)\cup\supp(\vect h_i^Z)\subseteq[n]$. Therefore,
$|\supp(\vect h_i^X)\cup\supp(\vect h_i^Z)|\le 1+\min\{n-1,\dim(\mathcal C_X+\mathcal C_Z)\}.$ 
Set $r_0=\min\{n-1,\dim(\mathcal C_X+\mathcal C_Z)\}.$
The codewords constructed above satisfy the paired-support condition in \eqref{eq.condition} with parameter $r_0$. 
Hence, $\mathcal Q(\mathcal C_X,\mathcal C_Z)$ has locality at most $r_0$.
Since $r$ denotes its minimum locality, we obtain
$$r\le r_0=\min\{n-1,k_X+k_Z-\dim(\mathcal C_X\cap\mathcal C_Z)\}.$$
This completes the proof.
\end{proof}

Note that Theorem~\ref{thm.eacss} and Corollary~\ref{cor.symmetric-eacss} 
allow us to derive the locality of a CSS-like EAQLRC from local properties of the underlying classical codes.  
Since they provide only sufficient conditions for local recovery, the actual
minimum locality of the resulting EAQLRC may be smaller than the guaranteed value.
In what follows, we call a locality provided by either
Theorem~\ref{thm.eacss} or Corollary~\ref{cor.symmetric-eacss}
\emph{trivial} if it coincides with the upper bound in
Theorem~\ref{thm.elementary-locality}, and \emph{nontrivial} if it is
strictly smaller than this bound.

\begin{remark}\label{rem.trivial-MDS}
Let $\mathcal C$ be an $[n,k,n-k+1]_q$ MDS code with
$1\leq k\leq n-1$ and an $\ell$-dimensional Euclidean hull.
Note that $k=n$ is meaningless for locality since
$\C^\perp=\{\vect 0\}$ in this case.
Taking $\mathcal C_X=\mathcal C_Z=\mathcal C$, we have
$\dim(\mathcal C_X\cap\mathcal C_Z)=k$.
By Corollary~\ref{cor.symmetric-eacss}, we can get a pure CSS-like
EAQLRC $\mathcal Q(\C,\C)$ with parameters
$\llbracket n,k-\ell,n-k+1;n-k-\ell\rrbracket_q$ and locality $k$.
Furthermore, Theorem~\ref{thm.elementary-locality} gives
$$k=\min\{n-1,k+k-\dim(\mathcal C_X\cap\mathcal C_Z)\}=\min\{n-1,k\}.$$
This shows that the locality derived from Corollary~\ref{cor.symmetric-eacss} attains 
the upper bound in Theorem~\ref{thm.elementary-locality}.
According to the above convention, such EAQLRCs are regarded as trivial for quantum local recovery.
\end{remark}

To describe the dependence among all parameters of CSS-like EAQLRCs, 
the following theorems present a Singleton-like bound 
and establish a criterion for such EAQLRCs being optimal with respect to this bound.  
The techniques are combinations of those used in \cite[Theorem~16]{LuoEzermanGrasslLing2024} and 
\cite[Theorem~4]{LuoChenEzermanLing2025}.

\begin{theorem}\label{thm.lrc-bounds}
Let $\mathcal Q$ be an $\dsb{n,\kappa,\delta;c}_q$ CSS-like EAQLRC with locality $r$ 
obtained from Theorem \ref{thm.eacss}.   
Then
\begin{equation}\label{eq:singleton-like} 2\delta \le n+c-\kappa-2\left\lceil\frac{\kappa}{r}\right\rceil+4. \end{equation}
\end{theorem}
\begin{proof}
By Theorem \ref{thm.eacss}, there exist two linear codes $\mathcal C_X$ and
$\mathcal C_Z$ with respective parameters
$[n, k_X, d_X]_q$ and $[n, k_Z, d_Z]_q$ such that $c=n-k_X-\ell_Z=n-k_Z-\ell_X$ and $\kappa=k_X+k_Z-n+c$, 
where $\ell_X=\dim(\mathcal C_X\cap \mathcal C_Z^\perp)$ and
$\ell_Z=\dim(\mathcal C_X^\perp\cap \mathcal C_Z)$.  
For every $i\in[n]$, Theorem~\ref{thm.eacss} also provides
$\vect h_i^X\in\mathcal C_X^\perp$ and
$\vect h_i^Z\in\mathcal C_Z^\perp$.  For each
$\alpha\in\{X,Z\}$, the paired support condition gives
$$i\in\supp(\vect h_i^\alpha)~{\rm and}~|\supp(\vect h_i^\alpha)| \le |\supp(\vect h_i^X)\cup\supp(\vect h_i^Z)| \le r+1.$$
Hence, Lemma~\ref{lem.classical-locality-dual} shows that both
$\mathcal C_X$ and $\mathcal C_Z$ have locality at most $r$.

The relations defining $c$ also give $\kappa=k_X-\ell_X=k_Z-\ell_Z$.
We consider separately the cases $\kappa=0$ and $\kappa>0$.

{\bf Case 1.} Assume that $\kappa=0$.  Then $\ell_X=k_X$ and
$\ell_Z=k_Z$, implying that
$\mathcal C_X\subseteq\mathcal C_Z^\perp$ and $\mathcal C_Z\subseteq\mathcal C_X^\perp$. 
Hence, \eqref{eq.eaqecc-css} gives $\delta=\min\{d_X,d_Z\}$. 
Applying Lemma \ref{lem.classical-singleton} to $\mathcal C_X$ and $\mathcal C_Z$ yields
$$2\delta \le d_X+d_Z \le 2n-k_X-k_Z -\left\lceil\frac{k_X}{r}\right\rceil -\left\lceil\frac{k_Z}{r}\right\rceil+4 = n+c -\left\lceil\frac{k_X}{r}\right\rceil -\left\lceil\frac{k_Z}{r}\right\rceil+4 \le n+c+4.$$
This is \eqref{eq:singleton-like} when $\kappa=0$. 

{\bf Case 2.} Assume that $\kappa>0$.  We can choose a generator matrix of
$\mathcal C_X$ whose first $\ell_X$ rows generate
$\mathcal C_X\cap\mathcal C_Z^\perp$. 
Note that permutation of coordinates does not change the locality of a code. 
After a coordinate permutation and elementary row operations, this generator
matrix can be written as
$$
  \begin{pmatrix}
  \mat{I}_{\ell_X} & \mat{R}_{\ell_X\times(n-\ell_X)} \\
  {\bf 0} & \mat{A}_X
  \end{pmatrix},
$$
where $\mat A_X$ has row rank $\kappa$.  Let $\mathcal D_X$ be the
$[n,\kappa]_q$ linear code generated by $({\bf 0}\ \mat A_X)$.
Since $\mathcal D_X\subseteq\mathcal C_X$, we have $\mathcal D_X\cap\mathcal C_Z^\perp = \mathcal D_X\cap \bigl(\mathcal C_X\cap\mathcal C_Z^\perp\bigr) =\{{\bf 0}\}$.
Therefore, $\mathcal D_X\setminus\{{\bf 0}\} \subseteq \mathcal C_X\setminus\mathcal C_Z^\perp \subseteq \mathcal C_X\setminus (\mathcal C_X\cap\mathcal C_Z^\perp)$.
With \eqref{eq.eaqecc-css}, we derive that 
$$d(\mathcal D_X) = \min_{\vect{x}\in\mathcal D_X\setminus\{{\bf 0}\}} \wt(\vect{x}) \ge \wt(\mathcal C_X\setminus (\mathcal C_X\cap\mathcal C_Z^\perp)) \ge\delta.$$

Since the first $\ell_X$ coordinates of every codeword in $\mathcal D_X$
are zero, puncturing these coordinates gives an
$[n-\ell_X=k_Z+c,\kappa,d_X'\ge\delta]_q$ code
$\mathcal D_X'$.  Since $\mathcal D_X\subseteq\mathcal C_X$, one has
$\mathcal C_X^\perp\subseteq\mathcal D_X^\perp$.  For every
$i\in[n]\setminus[\ell_X]$, deleting the first $\ell_X$ entries from the
local check $\vect h_i^X$ therefore gives a dual word of $\mathcal D_X'$
that contains the corresponding coordinate and has weight at most $r+1$.
Thus, $\mathcal D_X'$ has locality at most $r$.  Interchanging $X$ and $Z$ gives
a $[k_X+c,\kappa,d_Z'\ge\delta]_q$ code $\mathcal D_Z'$ with locality at
most $r$.
Again, applying Lemma \ref{lem.classical-singleton} to $\mathcal D_X'$ and $\mathcal D_Z'$, 
and using the fact $k_X+k_Z=\kappa+n-c$, we have 
$$2\delta\leq d_X'+d_Z'\le k_X+k_Z+2c-2\kappa -2\left\lceil\frac{\kappa}{r}\right\rceil+4 = n+c-\kappa -2\left\lceil\frac{\kappa}{r}\right\rceil+4,$$
which proves \eqref{eq:singleton-like}. 

Combining the above two cases completes the proof. 
\end{proof}

The Singleton-like bound in \eqref{eq:singleton-like} applies to every
CSS-like EAQLRC, whether pure or impure, but is not always attainable.
The following theorem presents necessary and sufficient conditions for the attainability of the bound in \eqref{eq:singleton-like} 
for pure CSS-like EAQLRCs.

\begin{theorem}\label{thm.singleton-equality}
Let $\mathcal C_X$ and $\mathcal C_Z$ be $[n,k_X,d_X]_q$ and
$[n,k_Z,d_Z]_q$ codes, respectively, such that
$(\mathcal C_X,\mathcal C_Z)$ satisfies the paired support conditions in \eqref{eq.condition}.  
Let $\mathcal Q(\mathcal C_X,\mathcal C_Z)$ be the pure $\dsb{n,\kappa,\delta;c}_q$ CSS-like EAQLRC of locality $r$ constructed
from Theorem \ref{thm.eacss}.
Then the EAQLRC $\mathcal Q(\mathcal C_X,\mathcal C_Z)$ is optimal with respect to the bound in
\eqref{eq:singleton-like}
if and only if the following three conditions hold:
\begin{enumerate}
  \item [\rm 1)]$\mathcal C_X$ and $\mathcal C_Z$ have the same parameters. 
  
  \item [\rm 2)]Both $\mathcal C_X$ and $\mathcal C_Z$ are optimal cLRCs meeting
  the Singleton-like bound in \eqref{eq.singleton-classical}. 

  \item [\rm 3)] $\left\lceil\frac{k_X}{r}\right\rceil=\left\lceil\frac{k_Z}{r}\right\rceil=\left\lceil\frac{\kappa}{r}\right\rceil$.
\end{enumerate}
\end{theorem}
\begin{proof}
Substituting the three conditions into \eqref{eq:singleton-like} yields
equality and proves sufficiency.  We next prove necessity.

Suppose that $\mathcal Q(\mathcal C_X,\mathcal C_Z)$ is optimal with respect to the bound in \eqref{eq:singleton-like}. 
Recall that we have shown in the proof of Theorem~\ref{thm.lrc-bounds} that 
both $\mathcal C_X$ and $\mathcal C_Z$ have locality at most $r$. 
Therefore, applying the classical Singleton-like LRC bound in \eqref{eq.singleton-classical} 
to $\mathcal C_X$ and $\mathcal C_Z$ gives 
$$d_X \le n-k_X-\left\lceil\frac{k_X}{r}\right\rceil+2~{\rm and}~ d_Z \le n-k_Z-\left\lceil\frac{k_Z}{r}\right\rceil+2.$$
Since $\delta=\min\{d_X,d_Z\}$ and $\kappa=k_X+k_Z-n+c\leq \min\{k_X,k_Z\}$, 
we obtain
\begin{equation}\label{eq:~pure-equality-chain}
\begin{aligned}
2\delta
&\le d_X+d_Z
\le 2n-k_X-k_Z-\left\lceil\frac{k_X}{r}\right\rceil
-\left\lceil\frac{k_Z}{r}\right\rceil+4\\
&\le 2n-k_X-k_Z-2\left\lceil\frac{\kappa}{r}\right\rceil+4
=n+c-\kappa-2\left\lceil\frac{\kappa}{r}\right\rceil+4.
\end{aligned}
\end{equation}

Since $\mathcal Q(\mathcal C_X,\mathcal C_Z)$ is optimal, the first and last terms in
\eqref{eq:~pure-equality-chain} are equal.  Hence, all the inequalities in this
chain must be equalities.
Note also that the following facts hold. 
\begin{itemize}
  \item [\rm (\romannumeral 1)] The first equality in \eqref{eq:~pure-equality-chain} holds if and only if $d_X=d_Z$. 

  \item [\rm (\romannumeral 2)] Combining the classical Singleton-like bound for cLRCs in \eqref{eq.singleton-classical} with the fact that $d_X=d_Z$, 
  the second equality in \eqref{eq:~pure-equality-chain} holds if and only if 
$d_X=n-k_X-\left\lceil\frac{k_X}{r}\right\rceil+2$ and $d_Z=n-k_Z-\left\lceil\frac{k_Z}{r}\right\rceil+2$.  
Equivalently, $\mathcal C_X$ and $\mathcal C_Z$ are optimal cLRCs meeting 
the classical Singleton-like bound in \eqref{eq.singleton-classical}. 
This gives condition {\rmfamily 2)}.

  \item [\rm (\romannumeral 3)] The third equality in \eqref{eq:~pure-equality-chain} holds if and only if $\left\lceil\frac{k_X}{r}\right\rceil + \left\lceil\frac{k_Z}{r}\right\rceil = 2\left\lceil\frac{\kappa}{r}\right\rceil$.
Since $\kappa\le \min\{k_X,k_Z\}$, we have $\left\lceil\frac{k_X}{r}\right\rceil = \left\lceil\frac{k_Z}{r}\right\rceil = \left\lceil\frac{\kappa}{r}\right\rceil$.
This gives condition {\rmfamily 3)}. 

\item [\rm (\romannumeral 4)] By {\rm (\romannumeral 1)} and
{\rm (\romannumeral 2)}, we deduce that $k_X+\left\lceil\frac{k_X}{r}\right\rceil = k_Z+\left\lceil\frac{k_Z}{r}\right\rceil.$ 
The equality of the ceiling terms in {\rm (\romannumeral 3)} therefore yields 
$k_X=k_Z$.  Together with $d_X=d_Z$, this proves that
$\mathcal C_X$ and $\mathcal C_Z$ have the same parameters, giving
condition {\rmfamily 1)}.
\end{itemize}
This completes the proof.
\end{proof}

\section{Two explicit families of optimal pure CSS-like EAQLRCs}\label{sec:constructions}

In this section, we develop a framework for constructing optimal pure  CSS-like EAQLRCs 
with respect to the Singleton-like bound in \eqref{eq:singleton-like} 
and apply this framework to obtain two explicit families of optimal pure CSS-like EAQLRCs. 
We also provide two concrete examples to illustrate our constructions.

\subsection{A framework}\label{subsec:framework}

By combining Theorem~\ref{thm.eacss} and Lemmas \ref{lem.classical-locality-dual} and \ref{lem:CSS-construction}, 
the problem of constructing optimal CSS-like EAQLRCs with respect to the Singleton-like bound 
can be reduced to the problem of designing classical parity-check matrices of $\C_X$ and $\C_Z$.   
We have a framework with the following three steps for constructing an optimal pure  $\dsb{n,\kappa,d;c}_q$ EAQLRC with locality $r$ 
with respect to the Singleton-like bound in \eqref{eq:singleton-like}: 
\begin{itemize}
  \item {\bf Step 1:} Choose two parity-check matrices $\mat H_X$ and
$\mat H_Z$ such that the corresponding codes $\mathcal C_X$ and
$\mathcal C_Z$ are optimal $[n,k,d]_q$ cLRCs with locality $r$ and
satisfy the paired support condition in \eqref{eq.condition}.  

  This step guarantees that the resulting CSS-like EAQLRC $\mathcal{Q}(\C_X,\C_Z)$ has locality $r$ and the conditions 1) and 2) 
  in Theorem~\ref{thm.singleton-equality} hold. 
  In particular, there are many candidates for such parity-check matrices in the literature where optimal cLRCs were constructed. 

  \item {\bf Step 2:} Choose a nonsingular $n\times n$ diagonal matrix $\mat D$ over $\F_q$ such that 
  $$\left\lceil\frac{k}{r}\right\rceil=\left\lceil\frac{2k-n+\rank(\mat H_X\mat D\mat H_Z^{\top})}{r}\right\rceil,$$  
  where $\kappa=2k-n+c$ with $c=\rank(\mat H_X\mat D\mat H_Z^{\top})$ 
  by Lemmas \ref{lem:CSS-construction} and \ref{lem.defect}. 

  It is clear that each row of $H_Z$ has the same support as the corresponding row of $H_ZD$. 
  Let $\C_Z'$ be the code with parity-check matrix $\mat H_Z\mat D$.
  Then $\C_Z$ and $\C_Z'$ have the same parameters and locality. 
  Therefore, $\C_Z'$ is an optimal $[n,k,d]_q$ cLRC with locality $r$ if and only if 
  $\C_Z$ is an optimal $[n,k,d]_q$ cLRC with locality $r$. 
  As a result, the CSS-like EAQLRC $\mathcal{Q}(\C_X,\C_Z')$ has locality $r$ and all the conditions 1)-3) 
  in Theorem~\ref{thm.singleton-equality} hold. 
  In addition, the introduction of $\mat{D}$ may improve the possibility for making condition 3) hold 
  compared to directly using $\mat H_X$ and $\mat H_Z$ as parity-check matrices of $\C_X$ and $\C_Z$ constructed in the literature.

\item {\bf Step 3:} Show that $\mathcal{Q}(\C_X,\C_Z')$ is pure. 

From Theorem~\ref{thm.singleton-equality}, we finally obtain a 
pure 
$$\dsb{n,2k-n+\rank(\mat H_X\mat D\mat H_Z^{\top}),d;\rank(\mat H_X\mat D\mat H_Z^{\top})}_q$$ 
EAQLRC with locality $r$, which is optimal with respect to the Singleton-like bound in \eqref{eq:singleton-like}. 
\end{itemize}

\begin{remark} We have the following remarks about the above framework. 
  \begin{enumerate}
    \item  [\rm 1)]   The introduction of the nonsingular diagonal matrix $\mat D$ in the above framework 
    has at least the following three effects. 
    \begin{itemize}
        \item Compared with directly using parity-check matrices $\mat{H}_X$ and $\mat{H}_Z$ 
        from existing constructions of optimal cLRCs $\mathcal C_X$ and $\mathcal C_Z$, 
        the freedom to choose $\mat{D}$ can make condition~{\rm 3)} easier to satisfy. 

      \item Since $\mat{D}$ is nonsingular and diagonal, multiplication by $\mat{D}$ rescales 
      the coordinates without changing their zero positions. 
      Consequently, the paired support condition in \eqref{eq.condition} 
      is preserved for the new pair $(\mathcal C_X,\mathcal C_Z')$.

      \item Different choices of $\mat{D}$ may change the numbers of the maximally entangled pairs and dimensions of 
      the resulting EAQLRCs while keeping their code lengths, minimum distances, and localities. 
      This is important for constructing EAQLRCs with flexible parameters. 
    \end{itemize}

    \item  [\rm 2)] Anderson $et~al.$ \cite{Anderson2022} proved that, for $q>2$, 
monomial equivalence can reduce the relative hull dimension. Their result guarantees 
a monomial transformation, which may contain a coordinate permutation, whereas the 
existence of a suitable diagonal transformation is not guaranteed, as emphasized in 
\cite[Remark~3.10]{Anderson2022}. Although monomial equivalence preserves the individual 
locality of each classical code, a permutation applied only to $\mathcal C_Z$ may misalign 
its local dual checks with those of the fixed code $\mathcal C_X$ and hence need not preserve 
the paired support condition in \eqref{eq.condition}. Therefore, the results in 
\cite{Anderson2022} cannot be directly applied to our framework. 
  \end{enumerate}
\end{remark}

In the following subsections, we apply the above framework to construct optimal EAQLRCs. 
Specifically, we choose $\mat H_X$ and $\mat H_Z$ to be parity-check matrices of previously well-constructed optimal cLRCs
and find appropriate nonsingular diagonal matrices $\mat D$ described in the framework.

\subsection{Optimal EAQLRCs from \texorpdfstring{$\ell$}{ell}-intersection pairs of MDS codes}

Recall that Remark \ref{rem.trivial-MDS} shows that a single MDS code can 
be used to construct a pure CSS-like EAQLRC, while its locality achieves the maximum possible. 
For quantum local recovery, such CSS-like EAQLRCs are regarded as trivial. 
In this subsection, we set $\mat H_X$ and $\mat H_Z$ to be parity-check matrices of two different MDS codes, 
and take $\mat{D}=\mat I_n$ in the framework of Subsection~\ref{subsec:framework}. 
Note that each $[n,k,n-k+1]_q$ MDS code is an optimal cLRC with locality $r=k$. 
As a result, we can obtain a family of optimal pure CSS-like EAQLRCs 
with flexible parameters and nontrivial locality. 
To this end, we need the following lemma.

\begin{lemma}\label{lem.grs-intersection-pair}
{\rm (\!\! \cite[Proposition~1, Remark~1, and Theorems~2 and 3]{HuangFangFu2022Intersection})}
Let $q\ge3$ be a prime power, let $n,k_1,k_2>0$ and $\ell\geq 0$ be 
integers satisfying $n\le q+1$, $k_1,k_2\le n-1$, $\max\{k_1+k_2-n,0\}\le\ell\le\min\{k_1,k_2\}$,
$(n,k_1,k_2)\notin\{(q+1,\ell,\ell+1),(q+1,\ell+1,\ell)\}$,
and $(n,\ell,\min\{k_1,k_2\})\ne(q+1,0,1)$.
Then there exist an $[n,k_1,n-k_1+1]_q$ MDS code $\mathcal C_1$ and an $[n,k_2,n-k_2+1]_q$ MDS code $\mathcal C_2$ such that
$\dim(\mathcal C_1\cap\mathcal C_2)=\ell$.
\end{lemma}

\begin{theorem}\label{thm.mds-eaqlrc-spectrum}
Let $q\ge3$ be a prime power, and let $d,n$ be positive integers satisfying
$2\le d\le n\le q+1$. For every integer $c$ satisfying
$\max\{0,2d-n-1\}\le c\le d-1$ and
$(n,d,c)\notin\{(q+1,2,1),(q+1,q+1,q)\}$,
with the additional restriction
$(n,d,c)\ne(q+1,q/2+1,0)$ when $q$ is even,
there exists a pure
$\dsb{n,n+c-2d+2,d;c}_q$ EAQLRC with locality $r=n-d+1$ that is optimal
with respect to the Singleton-like bound in \eqref{eq:singleton-like}.
\end{theorem}
\begin{proof}
\noindent{\bf Step 1.}
Set $k_1=d-1$, $k_2=n-d+1$, and $\ell=d-1-c$. Then
$k_1+k_2=n$, and the assumed range of $c$ gives
\begin{align}\label{eq.conA.1}
0\le\ell\le\min\{d-1,n-d\}\le\min\{k_1,k_2\}.
\end{align}
Moreover, $1\le k_1,k_2\le n-1$. 
Under the above substitution, the condition
$(n,d,c)\notin\{(q+1,2,1),(q+1,q+1,q)\}$ excludes the case
$(n,\ell,\min\{k_1,k_2\})=(q+1,0,1)$. 
By \eqref{eq.conA.1}, $\ell\le k_2-1$, implying that 
$(k_1,k_2)=(\ell+1,\ell)$ is impossible.
The additional condition $(n,d,c)\ne(q+1,q/2+1,0)$, 
imposed when $q$ is even, also excludes the remaining case $(n,k_1,k_2)=(q+1,\ell,\ell+1)$. 
Therefore, it follows from Lemma~\ref{lem.grs-intersection-pair} that there exist 
two MDS codes $\mathcal C_1$ and $\mathcal C_2$ with respective parameters
$[n,d-1,n-d+2]_q$ and $[n,n-d+1,d]_q$, respectively, such that $\dim(\mathcal C_1\cap\mathcal C_2)=\ell$. 
Let $\mat H_X$ be a generator matrix of $\mathcal C_1$, and let
$\mat H_Z$ be a generator matrix of $\mathcal C_2^\perp$.  Equivalently,
$\mat H_X$ and $\mat H_Z$ are respective parity-check matrices of
$$
\mathcal C_X=\mathcal C_1^\perp
~{\rm and}~
\mathcal C_Z=\mathcal C_2,
$$
where both $\mathcal C_X$ and $\mathcal C_Z$ are optimal $[n,k=n-d+1,d]_q$ cLRCs with locality $r=n-d+1$.

Since the dual of an MDS code is also MDS, both
$\mathcal C_X^\perp=\mathcal C_1$ and
$\mathcal C_Z^\perp=\mathcal C_2^\perp$ are
$[n,d-1,n-d+2]_q$ MDS codes.  For each $i\in[n]$, choose
$S_i\subseteq[n]$ such that $i\in S_i$ and
$|S_i|=n-d+2=r+1$.  Since every set of $n-d+2$ coordinates is the
support of a minimum weight codeword in an
$[n,d-1,n-d+2]_q$ MDS code, there exist
$\vect h_i^X\in\mathcal C_X^\perp$ and
$\vect h_i^Z\in\mathcal C_Z^\perp$ such that
$$
i\in \supp(\vect h_i^X)=\supp(\vect h_i^Z)=S_i.
$$
Hence, the paired support condition in \eqref{eq.condition} holds.

\noindent{\bf Step 2.}
Take $\mat D=\mat I_n$. Let $\C_Z'$ be the code with parity-check matrix $\mat H_Z'=\mat H_Z\mat D=\mat H_Z$. 
Then $\C_Z'=\C_Z$.
By Lemma \ref{lem.defect}, we have 
\begin{equation}\label{eq.conA.2}
\begin{split}
\rank(\mat H_X\mat D\mat H_Z^{\top})
= & \rank(\mat H_X\mat H_Z^{\top}) \\
= & n-k_X-\dim(\C_X^\perp\cap\C_Z)
=k_1-\dim(\mathcal C_1\cap\mathcal C_2)
=d-1-\ell
=c.
\end{split}
\end{equation}
Combining \eqref{eq.conA.1} and \eqref{eq.conA.2}, we verify that 
$$
\left\lceil
\frac{2k-n+\rank(\mat H_X\mat D\mat H_Z^{\top})}{r}
\right\rceil
=\left\lceil\frac{n-d-\ell+1}{n-d+1}\right\rceil
=1
=\left\lceil\frac{k}{r}\right\rceil
$$
holds. 

\noindent{\bf Step 3.}
It remains to prove purity. 
On the one hand, it follows from Theorem \ref{thm.eacss} that 
$\mathcal{Q}(\mathcal C_X,\mathcal C_Z')$ has parameters 
$\dsb{n,n+c-2d+2,\delta\geq d;c}_q$ and locality $r=n-d+1$. 
On the other hand, the Singleton-like bound in Theorem~\ref{thm.lrc-bounds} yields 
$$
2\delta
\le n+c-\kappa
-2\left\lceil\frac{\kappa}{r}\right\rceil+4
=2d.
$$ 
Hence, $\delta=d$ and the CSS-like EAQLRC $\mathcal{Q}(\mathcal C_X,\mathcal C_Z')$ is pure.

By the framework of Subsection~\ref{subsec:framework}, we conclude that  
$\mathcal{Q}(\mathcal C_X,\mathcal C_Z')$ is a pure 
$\dsb{n,n+c-2d+2,d;c}_q$ EAQLRC with locality $r=n-d+1$ that is optimal
with respect to the Singleton-like bound in \eqref{eq:singleton-like}. 
This completes the proof. 
\end{proof}

The following remark and example indicate that the optimal EAQLRCs constructed in Theorem~\ref{thm.mds-eaqlrc-spectrum} can have nontrivial locality.

\begin{remark}\label{rem.con1}
From the proof of Theorem~\ref{thm.mds-eaqlrc-spectrum}, the two underlying classical codes $\C_X$ and $\C_Z'$ 
have the same dimension $k=n-d+1$ and locality $r=k$. 
By Theorem~\ref{thm.elementary-locality}, the corresponding upper bound on the locality of $\mathcal{Q}(\mathcal C_X,\mathcal C_Z')$ is
$\min\{n-1,2k-\dim(\mathcal C_X\cap\mathcal C_Z')\}$. 
Since $\dim(\mathcal C_X\cap\mathcal C_Z')\le k$ and $k\le n-1$, both terms in this minimum are at least $k$.  
Therefore, the constructed locality $r=k$ attains this upper bound if and only if
$n-1=k$ or $\dim(\mathcal C_X\cap\mathcal C_Z')=k$, 
which is further equivalent to $d=2$ or $\mathcal C_X=\mathcal C_Z'$. 
However, $d>2$ and $\mathcal C_X\neq \mathcal C_Z'$ are possible in the construction of Theorem~\ref{thm.mds-eaqlrc-spectrum}. 
This fact implies that optimal EAQLRCs obtained in Theorem~\ref{thm.mds-eaqlrc-spectrum} can also have nontrivial locality 
even if we employ MDS codes as the underlying classical codes. 
\end{remark}

\begin{example}\label{exam.1}
Let $q=5$ and define
$$
\mat H_X=
\begin{pmatrix}
1&1&1&1&1\\
0&1&2&3&4
\end{pmatrix}
~{\rm and}~
\mat H_Z=
\begin{pmatrix}
1&1&2&3&4\\
0&1&4&4&1
\end{pmatrix}.
$$
In the framework of Subsection~\ref{subsec:framework}, set $\mat D=\mat I_5$.  Then
$$
\mat H_Z'=\mat H_Z\mat D=
\begin{pmatrix}
1&1&2&3&4\\
0&1&4&4&1
\end{pmatrix}.
$$
Set $\mathcal C_1=\row(\mat H_X)$ and $\mathcal C_2=\ker(\mat H_Z)$. 
Then $\mathcal C_1$ and $\mathcal C_2$ are $[5,2,4]_5$ and $[5,3,3]_5$ MDS codes, respectively. 
A direct calculation gives
$\dim(\mathcal C_1\cap\mathcal C_2)=1$. 
Thus, $(\mathcal C_1,\mathcal C_2)$ is a $1$-intersection pair corresponding to
$(q,n,d,c)=(5,5,3,1)$ in Theorem~\ref{thm.mds-eaqlrc-spectrum}. 
With the proof of Theorem~\ref{thm.mds-eaqlrc-spectrum}, 
we take $\mathcal C_X=\mathcal C_1^\perp$ and
$\mathcal C_Z'=\mathcal C_2$ to be both $[5,3,3]_5$ MDS codes.  
By Theorem~\ref{thm.mds-eaqlrc-spectrum}, we immediately obtain 
an optimal pure $\dsb{5,2,3;1}_5$ EAQLRC $\mathcal{Q}(\mathcal C_X,\mathcal C_Z')$ with locality $r=n-d+1=3$.

We also verify that $\dim(\mathcal C_X\cap\mathcal C_Z')=1$. 
Consequently, the upper bound in Theorem~\ref{thm.elementary-locality} gives 
$$\min\{n-1,k_X+k_Z-\dim(\mathcal C_X\cap\mathcal C_Z')\}
=\min\{4,3+3-1\}=4.$$  
Since $r=3<4$, the optimal EAQLRC $\mathcal{Q}(\mathcal C_X,\mathcal C_Z')$ has a nontrivial locality.
\end{example}

\subsection{Optimal EAQLRCs from block parity-check matrices}

In this subsection, we construct optimal EAQLRCs using the single matrix specialization
of the framework in Subsection~\ref{subsec:framework}, namely setting $\mat H_X=\mat H_Z=\mat H$. 
In this case, the nonsingular diagonal matrices $\mat D$ are generally not the identity matrix $\mat I_n$.  
Note also that there is a systematic construction of optimal cLRCs of length 
$n$ and locality $r$ such that $r+1\mid n$ by using block parity-check matrices with components that are 
parity-check matrices of generalized Reed-Solomon (GRS) codes, as summarized in \cite[Lemma~10]{LuoChenEzermanLing2025}. 

We first recall the definition of GRS codes.  
Let $\boldsymbol{a}=(a_1,a_2,\ldots,a_r)$ consist of pairwise distinct
elements of $\F_q$, and let $\vect v=(v_1,v_2,\ldots,v_r)\in(\F_q^*)^r$.  For
$1\le k\le r\leq q$, the {\em generalized Reed-Solomon (GRS) code} 
$$
{\rm GRS}_k(\boldsymbol{a},\vect v)
=\{(v_1f(a_1),v_2f(a_2),\ldots,v_rf(a_r)):~f\in\F_q[x]~{\rm with}~\deg(f)<k\}
$$ 
is an $[r,k,r-k+1]_q$ MDS code with generator matrix
$$
\mat G_{k,{\bf a},{\bf v}}=\begin{pmatrix}
v_1&v_2&\cdots&v_r\\
v_1a_1&v_2a_2&\cdots&v_ra_r\\
\vdots&\vdots&\ddots&\vdots\\
v_1a_1^{k-1}&v_2a_2^{k-1}&\cdots&v_ra_r^{k-1}
\end{pmatrix}.
$$

Let $q$ be a prime power, and let $d,u,r$ be positive integers satisfying
$3\le d\le r+1\le q$. For each $i\in[u]$, choose
$\vect a_i=(a_{i,1},a_{i,2},\ldots,a_{i,r+1})\in\F_q^{r+1}$ with pairwise distinct
coordinates and $\vect v_i\in(\F_q^*)^{r+1}$, and write
\begin{equation}\label{eq:~grs-block-check111}
  \mat G_{d-1,\vect a_i,\vect v_i}
  =
  \begin{pmatrix}
    \vect v_i\\
    \mat G_i
  \end{pmatrix}
  ~{\rm and}~
  A_i=\{a_{i,1},a_{i,2},\ldots,a_{i,r+1}\}.
\end{equation}
Let 
$\mat V={\rm diag}(\vect v_1,\vect v_2,\ldots,\vect v_u)~{\rm and}~\mat G=(\mat G_1\ \cdots\ \mat G_u)$. 
Consider the block matrix 
\begin{equation}\label{eq:~grs-block-check222}
  \mat H=
  \begin{pmatrix}
    \mat V\\
    \mat G
  \end{pmatrix}
  =
  \begin{pmatrix}
    \vect v_1       & \vect 0_{r+1} & \cdots & \vect 0_{r+1}\\
    \vect 0_{r+1}   & \vect v_2     & \cdots & \vect 0_{r+1}\\
    \vdots           & \vdots         & \ddots & \vdots\\
    \vect 0_{r+1}   & \vect 0_{r+1} & \cdots & \vect v_u\\
    \mat G_1         & \mat G_2       & \cdots & \mat G_u
  \end{pmatrix}.
\end{equation}

\begin{lemma}{\rm (\!\!\cite[Lemma~10]{LuoChenEzermanLing2025})}\label{lem.grs-block-classical-lrc}
Let $q$ be a prime power and let $u,d,r$ be positive integers such that $1\leq d-2<r\le q-1$.  
Let $\C$ be a linear code with parity-check matrix $\mat H$ in \eqref{eq:~grs-block-check222}. 
Then $\mathcal C$ is an optimal $[u(r+1),ur-d+2,d]_q$ linear code with locality $r$ if one of the following
conditions holds.
\begin{enumerate}
  \item [\rm 1)] $d\in\{3,4\}$, $\vect a_1=\cdots=\vect a_u$, and
  $\vect v_1=\cdots=\vect v_u$.
\item [\rm 2)] $d\ge5$ and $\left|\bigcup_{i\in S}A_i\right|\ge r|S|+1$, 
with $A_i$ as defined in \eqref{eq:~grs-block-check111}, 
for every $S\subseteq[u]$ with $1\le |S|\le\lfloor(d-1)/2\rfloor$.
\end{enumerate}
\end{lemma}

\begin{theorem}\label{thm.grs-coset-flexible}
Let $d,u,r$ be positive integers with $r>2d-4\ge 2$, $(r+1)\mid (q-1)$, and $r+1<q-1$.  
For each integer $s$ satisfying $0\le s\le d-2$, if 
one of the following conditions holds: 
\begin{enumerate}
  \item [\rm 1)] $d\in\{3,4\}$;
  \item [\rm 2)] $d\ge5$ and $u(r+1)\le q-1$, 
\end{enumerate}
then there exists a pure $\dsb{u(r+1),ur-2d+4+s,d;u+s}_q$ EAQLRC with locality $r$, 
which is optimal with respect to the Singleton-like bound in \eqref{eq:singleton-like}. 
\end{theorem}
\begin{proof}
\noindent{\bf Step 1.}
Let $G$ be a subgroup of $\F_q^*$ of order $r+1$.  
According to the prescribed integer $d$, we consider two cases.  
\begin{itemize}
    \item If $d\in\{3,4\}$, we fix an ordering
    $G=\{a_1,a_2,\ldots,a_{r+1}\}$, and take
    $A_i=G$ and
    $\vect a_i=(a_1,a_2,\ldots,a_{r+1})$ for every $i\in[u]$.
    Thus, $a_{i,j}=a_j$ for all $i\in[u]$ and $j\in[r+1]$.

    \item If $d\ge5$ and $u(r+1)\le q-1$, we choose pairwise distinct
    cosets $A_i=\beta_iG$ and fix an ordering
    $A_i=\{a_{i,1},a_{i,2},\ldots,a_{i,r+1}\}$ for each $i\in[u]$.
    Take $\vect a_i=(a_{i,1},a_{i,2},\ldots,a_{i,r+1})$ for every
    $i\in[u]$. 
    Moreover, since the sets $A_i$ are pairwise disjoint, then, for every
nonempty $S\subseteq[u]$, we have 
$
\left|\bigcup_{i\in S}A_i\right|
=(r+1)|S|
\ge r|S|+1.
$
\end{itemize} 
For either case, we further take
$
\vect v_1=\vect v_2=\cdots=\vect v_u=(1,1,\ldots,1)\in(\F_q^*)^{r+1}
$ in \eqref{eq:~grs-block-check222}, 
so that
$
\mat G_i
=
\bigl(a_{i,j}^{\,t}\bigr)_{
\substack{1\le t\le d-2,~1\le j\le r+1}}.
$
Let $\mathcal C_X=\C_Z=\C$ be the linear code with parity-check matrix $\mat H$ in 
\eqref{eq:~grs-block-check222}.  
Moreover, the inequality $r>2d-4$ implies $r>d-2$, and the assumptions on $q$ give
$r\le q-1$.  
It then follows from Lemma~\ref{lem.grs-block-classical-lrc} that
$\mathcal C_X=\C_Z=\C$ is an optimal
$[u(r+1),ur-d+2,d]_q$ cLRC with locality $r$.

For each $i\in[u]$, let $\vect b_i$ denote the $i^{\rm th}$ row of the
block-diagonal matrix $\mat V$ in \eqref{eq:~grs-block-check222}.  Then
$\supp(\vect b_i)=J_i=\{(i-1)(r+1)+1,(i-1)(r+1)+2,\ldots,i(r+1)\}$ with $|J_i|=r+1$ and 
the sets $J_1,J_2,\ldots,J_u$ form a partition of $[u(r+1)]$.
For every coordinate $j$, let $i(j)$ be the unique index such that
$j\in J_{i(j)}$ and take
$
\vect h_j^X=\vect h_j^Z=\vect b_{i(j)}.
$
It then follows that
$$
j\in\supp(\vect h_j^X)\cap\supp(\vect h_j^Z)
~{\rm and}~
\left|\supp(\vect h_j^X)\cup\supp(\vect h_j^Z)\right|=r+1.
$$
Thus, the paired support condition in \eqref{eq.condition} holds.

\noindent{\bf Step 2.}
For every integer $s$ with $0\le s\le d-2$, put
$
e_s=2d-s-3.
$
Since $r>2d-4$, we have
$
d-2<e_s\le2d-3\le r<r+1.
$
Moreover, $|A_1|=r+1<q-1$.  Hence, for each
$0\le s\le d-2$, we can choose an element 
$
\theta_s\in\F_q^*\setminus
\{-a^{e_s}:~a\in A_1\}.
$
With the ordering
$A_1=\{a_{1,1},a_{1,2},\ldots,a_{1,r+1}\}$ fixed in {\bf Step~1}, 
we define the $ (u(r+1))\times(u(r+1)) $ diagonal matrix 
\begin{align}\label{eq.D_22}
\mat D_s
=
{\rm diag}\left(
1+\theta_s a_{1,1}^{-e_s},
1+\theta_s a_{1,2}^{-e_s},
\ldots,
1+\theta_s a_{1,r+1}^{-e_s},
\underbrace{1,1,\ldots,1}_{(u-1)(r+1)}
\right).
\end{align}
In particular, the choice of $\theta_s$ guarantees that $\mat D_s$ is nonsingular.

For every coset $A_i=\beta_iG$ and every integer $e$ with
$|e|<r+1$, we have
\begin{equation}\label{eq:~coset-power-sums}
\sum_{a\in A_i}a^e
=
\begin{cases}
r+1, & e=0,\\
0,   & 0<|e|<r+1.
\end{cases}
\end{equation}
It follows from \eqref{eq:~coset-power-sums} that, for
$0\le e\le2d-4$,
\begin{equation}\label{eq:~scaled-coset-power-sums}
\sum_{a\in A_1}
\bigl(1+\theta_s a^{-e_s}\bigr)a^e
=
\begin{cases}
r+1,           & e=0,\\
\theta_s(r+1), & e=e_s,\\
0,             & \text{otherwise}.
\end{cases}
\end{equation}
When $s=0$, we have $e_0=2d-3>2d-4$, and hence the second case in
\eqref{eq:~scaled-coset-power-sums} does not occur in this range.

Since
$
\mat H=
\begin{pmatrix}
\mat V\\
\mat G
\end{pmatrix},
$
it can be checked that 
$
\mat H\mat D_s\mat H^{\top}
=
\begin{pmatrix}
\mat V\mat D_s\mat V^{\top} &
\mat V\mat D_s\mat G^{\top}\\
\mat G\mat D_s\mat V^{\top} &
\mat G\mat D_s\mat G^{\top}
\end{pmatrix}.
$
From \eqref{eq:~coset-power-sums} and
\eqref{eq:~scaled-coset-power-sums}, we deduce that 
$\mat V\mat D_s\mat V^{\top}=(r+1)\mat I_u$ and 
$\mat V\mat D_s\mat G^{\top}
=\mat G\mat D_s\mat V^{\top}=\mat 0.
$
Moreover, for $1\le j,k\le d-2$, we have 
$$
(\mat G\mat D_s\mat G^{\top})_{j,k}
=
\begin{cases}
\theta_s(r+1), & j+k=e_s,\\
0,             & \text{otherwise}.
\end{cases}
$$
There are exactly $2(d-2)-e_s+1=s$ pairs $(j,k)$ satisfying $j+k=e_s$, 
and they occupy distinct rows and columns.  
In particular, there are no such pairs when $s=0$.  Since
$r+1\mid(q-1)$, the element $r+1$ is nonzero in $\F_q$.  Therefore,
$$
\rank(\mat H_X\mat D_s\mat H_Z^{\top})
=
\rank(\mat H\mat D_s\mat H^{\top})
=u+s.
$$

Let $\mathcal C_Z'$ be the code with parity-check matrix
$\mat H_Z\mat D_s$.  
Then both $\mathcal C_X$ and $\mathcal C_Z'$ are 
optimal $[n=u(r + 1), k=ur - d + 2, d]_q$ cLRCs with locality $r$. 
Furthermore, since $r>d-2\geq 1$ and $0<2d-4-s<r$, 
we have  
$$
\left\lceil\frac{k}{r}\right\rceil
=
\left\lceil u-\frac{d-2}{r}\right\rceil
=u
~{\rm and}~
\left\lceil
\frac{2k-n+\rank(\mat H_X\mat D_s\mat H_Z^{\top})}{r}
\right\rceil
=
\left\lceil u-\frac{2d-4-s}{r}\right\rceil
=u, 
$$
implying that $\left\lceil\frac{k}{r}\right\rceil=\left\lceil
\frac{2k-n+\rank(\mat H_X\mat D_s\mat H_Z^{\top})}{r}
\right\rceil
$ holds.

\noindent{\bf Step 3.}
It remains to prove purity. 
On the one hand, it follows from Theorem \ref{thm.eacss} that 
$\mathcal{Q}(\mathcal C_X,\mathcal C_Z')$ has parameters 
$\dsb{u(r+1),ur-2d+4+s,\delta\geq d;u+s}_q$ and locality $r$. 
On the other hand, the Singleton-like bound in Theorem~\ref{thm.lrc-bounds} yields 
\begin{align*}
2\delta\le
u(r+1)+u+s-(ur-2d+4+s)-2u+4
=2d.
\end{align*}
Therefore, $\delta=d$, and the CSS-like EAQLRC $\mathcal{Q}(\mathcal C_X,\mathcal C_Z')$ is pure.

By the framework of Subsection~\ref{subsec:framework}, we conclude that  
$\mathcal{Q}(\mathcal C_X,\mathcal C_Z')$ is a pure 
$\dsb{u(r+1),ur-2d+4+s,d;u+s}_q$ EAQLRC with locality $r$ that is
optimal with respect to the Singleton-like bound in \eqref{eq:singleton-like}. 
This completes the proof. 
\end{proof}

\begin{example}\label{exam.2}
Let $(q,d,r,u,s)=(13,3,3,2,1)$, which satisfies all the conditions in
Theorem~\ref{thm.grs-coset-flexible}. 
Let $G=\langle5\rangle=\{1,5,12,8\}$, which is a subgroup of $\F_{13}^*$ of order $r+1=4$. 
Let $\vect a_1=\vect a_2=(1,5,12,8)$ and let $\vect v_1=\vect v_2=(1,1,1,1)$.  
The block matrix in \eqref{eq:~grs-block-check222} gives
$$
\mat H_X=\mat H_Z=\mat H=
\begin{pmatrix}
1&1&1&1&0&0&0&0\\
0&0&0&0&1&1&1&1\\
1&5&12&8&1&5&12&8
\end{pmatrix}.
$$
For $s=1$, we have $e_s=2d-s-3=2$.  Set $\theta_s=2$, which is admissible in
the proof of Theorem~\ref{thm.grs-coset-flexible} since 
$\{-a^2:~a\in G\}=\{1,12\}$.  This choice yields
$\mat D={\rm diag}(3,12,3,12,1,1,1,1)$ in \eqref{eq.D_22} and hence
$$
\mat H_Z'=\mat H_Z\mat D=
\begin{pmatrix}
3&12&3&12&0&0&0&0\\
0&0&0&0&1&1&1&1\\
3&8&10&5&1&5&12&8
\end{pmatrix}.
$$
Let $\C_X$ and $\C_Z'$ be the codes with parity-check matrices $\mat H_X$ and $\mat H_Z'$, respectively.  
Then both $\C_X$ and $\C_Z'$ are optimal $[8,5,3]_{13}$ cLRCs with locality $3$.  

By Theorem~\ref{thm.grs-coset-flexible}, the code
$\mathcal{Q}(\mathcal C_X,\mathcal C_Z')$ is an optimal pure
$\dsb{8,5,3;3}_{13}$ EAQLRC with locality $r=3$.

We also verify that $\dim(\mathcal C_X\cap\mathcal C_Z')=3$.  Consequently, the
upper bound in Theorem~\ref{thm.elementary-locality} gives 
$$\min\{n-1,k_X+k_Z-\dim(\mathcal C_X\cap\mathcal C_Z')\}
=\min\{7,5+5-3\}=7.$$  
Since $r=3<7$, the optimal EAQLRC $\mathcal{Q}(\mathcal C_X,\mathcal C_Z')$ has a nontrivial locality. 
\end{example}

\section{Conclusion}\label{sec:~conclusion}

In this paper, we introduced EAQLRCs in Definition~\ref{def.eaqlrc} and established a sufficient criterion for EASCs to have locality $r$ in terms of the supports of their extended stabilizers in Theorem~\ref{thm.stabilizer-locality}.
Building on this criterion, we developed a CSS-like construction of EAQLRCs in Theorem~\ref{thm.eacss} and Corollary~\ref{cor.symmetric-eacss}.
For the resulting CSS-like EAQLRCs, we derived an upper bound on their minimum locality in Theorem~\ref{thm.elementary-locality} and a Singleton-like bound in Theorem~\ref{thm.lrc-bounds}.
Both bounds apply to arbitrary CSS-like EAQLRCs, whether pure or impure.
For pure CSS-like EAQLRCs, we further established necessary and sufficient conditions for attaining equality in the Singleton-like bound in Theorem~\ref{thm.singleton-equality}.
These results led to a general framework for constructing optimal pure CSS-like EAQLRCs from pairs of cLRCs in Subsection~\ref{subsec:framework}.
Finally, applying this framework to $\ell$-intersection pairs of MDS codes and block parity-check pairs, 
we constructed two families of optimal pure CSS-like EAQLRCs with flexible parameters and nontrivial localities, 
as stated in Theorems~\ref{thm.mds-eaqlrc-spectrum} and~\ref{thm.grs-coset-flexible}, respectively.

\end{sloppypar}
\end{document}